\pdfoutput=1
\RequirePackage{fix-cm}
\documentclass[11pt, copyright, gdm]{GDM/google}

\uselogo{} 
\usepackage[authoryear, sort&compress, round]{natbib}
\usepackage[parfill]{parskip}

\usepackage[T1]{fontenc}
\usepackage[utf8]{inputenc}
\usepackage{helvet}
\usepackage{courier}
\usepackage{microtype}

\usepackage{amssymb,amsthm,mathtools,bm}
\allowdisplaybreaks

\usepackage{booktabs}
\usepackage{tabularx}
\usepackage{multirow}
\usepackage{caption}
\usepackage{enumitem}
\usepackage{placeins}
\usepackage{comment}
\usepackage{etoc}

\usepackage{xcolor}
\usepackage{pgfplots}
\usepackage{pgfplotstable}
\usetikzlibrary{patterns}
\pgfplotsset{compat=1.18}
\usepgfplotslibrary{groupplots}
\usepackage[most]{tcolorbox}

\usepackage{todonotes}
\setuptodonotes{inline}

\usepackage{hyperref}
\hypersetup{
  colorlinks=true,
  linkcolor=blue,
  citecolor=blue,
  urlcolor=blue
}
\usepackage{cleveref}

\newtheoremstyle{mytheoremstyle}
  {.5em}       
  {.5em}       
  {\itshape}   
  {}           
  {\bfseries}  
  {.}          
  {.5em}       
  {}           

\theoremstyle{mytheoremstyle}
\newtheorem{theorem}{Theorem}
\newtheorem{lemma}[theorem]{Lemma}
\newtheorem{corollary}[theorem]{Corollary}
\newtheorem{proposition}[theorem]{Proposition}
\newtheorem{remark}[theorem]{Remark}
\newtheorem{definition}[theorem]{Definition}

\tcbset{
  basebox/.style={
    width=\linewidth,
    top=6pt,
    bottom=0pt,
    colframe=black,
    colbacktitle=black,
    enhanced,
    center,
    attach boxed title to top left={yshift=-0.1in, xshift=0.15in},
    boxed title style={boxrule=0pt, colframe=white}
  },
  aibox/.style={basebox, colback=blue!6!white},
  graybox/.style={basebox, colback=gray!20}
}
\newtcolorbox{AIbox}[2][]{aibox, title=#2, #1}
\newtcolorbox{Graybox}[2][]{graybox, title=#2, #1}

\renewcommand{\cite}{\citep}

\newcommand{\adam}[1]{\todo[inline,color=red]{AF: #1}}

\newcommand{\OPT}{\operatorname{OPT}}

\newcommand{\E}{\mathbb E}

\newcommand{\Prb}{\mathbb P}

\title{Price of Anarchy of First-Price Auctions in Agentic Markets}

\correspondingauthor{Correspondence to \texttt{\{shubhendut, jeisenstein, alekhagarwal, imgemp, mkaisers, fisch\}@google.com.}}

\author{%
  Shubhendu Trivedi$^1$, Jacob Eisenstein$^1$, Alekh Agarwal$^2$, Ian Gemp$^1$, Michael Kaisers$^1$, Adam Fisch$^1$ \\
$^1$Google DeepMind,
$^2$Google Research
}

\begin{document}
\begin{abstract}
An agentic market is a decentralized serving-time market where AI agents bid to claim each user prompt. An agent's bid is based on its expected value for the prompt, which is the expectation of the reward for its response minus its inference costs. First-price auctions (FPAs), in which the highest bidder pays its bid and receives the prompt, are a natural way to allocate prompts in such a market, since they are simple and transparent to run, and do not require central knowledge of private capabilities and costs. However, FPAs are not truthful: participants may bid less than their true value in the hope of a larger profit, and the prompt may then be allocated to someone other than the bidder who values it the most. This welfare loss is quantified by the \emph{price-of-anarchy} (PoA). In this paper, we argue that agentic markets are likely to possess features that reduce the price-of-anarchy. The first is a \emph{competitive floor}, which occurs when multiple bidders usually value the prompt above some baseline share of the optimal welfare (e.g., because the best models are often distilled into cheaper alternatives, which perform similarly). We show that competition among them then limits the welfare loss, which can be driven to zero as the floor rises. The second feature is \emph{near-independence} of values after conditioning on the prompt, which occurs when most of the predictable variation in the expected value is determined by the difficulty of the prompt. As the value distribution approaches conditional independence, we can derive an efficiency guarantee that approaches a factor of $1 - 1/e^2$ of the optimal welfare. Thus, both features give rise to welfare lower bounds that are significantly better than the well-known general case of $1 - 1/e$. We  validate these results empirically using auctions with strategic bidders simulated on the EmbedLLM prompt routing dataset.\looseness=-1
\end{abstract}
\maketitle
\section{Introduction}

Consider a user with a task that could be serviced by any one of many AI agents, where each AI agent has a different expected cost and response quality. The user could choose among the agents directly, but that requires knowledge about every option on the market. Alternatively, aggregators and routers can make this choice on the user's behalf---e.g., by establishing a central view of which agent is good at which tasks, and directing users' prompts accordingly~\citep{ding2024hybrid, chen2024frugalgpt, ong2025routellm, aubakirova2025stateofai}. That view has to be measured and maintained, however, which becomes harder as the candidate pool grows beyond well-known providers. A third, and emerging, approach is the "agentic market," which is a serving-time market where possible responding agents bid to claim each prompt, and are rewarded according to the quality of the response they provide~\citep{qi2026economymindsemergingmultiagent, chang2020decentralized, liu2026markets}. Each responder is assumed to privately know (or estimate) its expected quality and cost on the current prompt, which the market must elicit efficiently and transparently, without assuming that responders are honest. Being decentralized, the market also scales with the number of participants rather than with the platform's ability to track them.\looseness=-1

The main design choice in such a market is how to allocate prompts based on the agents' bids, and what price to charge the agent who wins the auction. Many  mechanisms exist, and recent work on agentic markets has largely favored those that are dominant-strategy incentive compatible, such as second-price and Vickrey-Clarke-Groves (VCG) variants, under which bidding one's true value is optimal regardless of what other bidders do~\citep{chang2020decentralized, liu2026markets, cao2026secondbest}.\looseness=-1

In this paper, we analyze the viability of  the first-price auction (FPA), which is an especially simple mechanism. Here, each responder bids for the right to service the current prompt. The highest bidder wins, pays its bid to the platform, submits its response, and collects a reward based on quality~\cite{qi2026economymindsemergingmultiagent, chang2020decentralized}. FPAs have several advantages. The payment rule is transparent, since the winner pays its own bid, with no need to trust the platform to truthfully report a second price~\citep{rothkopf1990vickrey}. Bidders can be strategic, but do not need to be: a nonnegative bid below expected quality minus cost results in a nonnegative expected surplus when it wins, regardless of what others bid. FPAs are also simple to implement, which partly explains why some of the largest markets, such as display advertising auctions, use them~\citep{googleadmanager2019firstprice, paes_leme_2020}.\looseness=-1

While FPAs are successful in practice, they are less appreciated in theory. Unlike VCG-style mechanisms, first-price auctions are not truthful: strategic agents ``shade'' their bids below their true values to increase profit, and if higher-valued agents shade too aggressively, they can lose to lower-valued agents, which decreases \emph{social welfare} (the total value generated by the allocation). The ratio of the welfare at a strategic equilibrium to the best achievable welfare is called the \emph{price of anarchy}~\citep{RoughgardenPOA}, and for FPAs it can be as low as $1 - 1/e\approx 63\%$~\citep{SyrgkanisTardos2013}.\looseness=-1

A potential ${1}/{e} \approx 37\%$ loss in efficiency may seem severe, but it is a worst case bound over arbitrary value distributions. 
We empirically show that in representative agentic markets simulated on the large-scale EmbedLLM dataset~\cite{zhuang2025embedllm}, realized welfare losses are often much smaller. To understand this gap between theory and practice, we derive new bounds that show that the worst case improves substantially under structural conditions that are natural in agentic markets.
The primary condition is a \emph{competitive floor}, which occurs when several bidders usually value the query above a certain baseline. Competition among these bidders then keeps the winning bid high and limits how much welfare can be lost. Competitive floor conditions can arise when agents have comparable substitutes, e.g., due to distillation or similar training data and capabilities~\citep{cao2026secondbest, aubakirova2025stateofai}. Our bounds depend on how many bidders are competitive, how high the floor is, and how reliably it holds. As a secondary result, we show that the PoA bound can also improve through \emph{near-independence}, even when no such floor exists. When bidder values are independent, \citet{JinLu2022} show that the price of anarchy improves to $1 - 1/e^2 \approx 86.5\%$. We prove that this guarantee still holds approximately under small perturbations from independence in terms of total-variation distance. This is practically relevant because shared query features (such as prompt difficulty or type) account for much of the co-movement across models; the idea being that once this information is conditioned out, the residual dependence may be much weaker.\looseness=-1

These results show that when models (1) compete closely on capability or (2) share task difficulty that can be conditioned out, the equilibrium allocation preserves substantially more welfare than worst-case theory predicts. 
Much of the motivation for more complex mechanisms then falls away---the simple and transparent FPA  can already capture most of the available welfare under conditions that agentic markets tend to satisfy. A market platform also does not have to treat these conditions as purely exogenous. In retail, for example, operators often stock ``private labels'' alongside third-party sellers to discipline markups~\citep{SHOPOVA2023102949, ding2025marketplaceoperatorsinducecompetitive}. In the same spirit, we examine how a platform can steer market conditions, either by introducing its own model to increase competition (acting like a reserve price) or by filtering out low-performing responders before bidding begins.\looseness=-1

\textbf{Contributions of this work.} In summary, this paper provides theoretical and empirical tools to analyze first-price auctions under what practitioners may expect from typical agentic market structure:
\begin{enumerate}[leftmargin=*]
    \item \textbf{Competitive floors (\S\ref{sec:concentration}):} We derive price-of-anarchy bounds that improve when several bidders usually value the query above a certain baseline, a fixed fraction of the optimal welfare.\looseness=-1
    \item \textbf{Near-independence (\S\ref{sec:near_independence}):} We show the $1 - 1/e^2$ bound for independent values~\citep{JinLu2022} still holds approximately when values are nearly independent after conditioning on the query.\looseness=-1
    \item \textbf{Platform design (\S\ref{sec:applications}):} We show that platforms can create these conditions, by filtering out weak responders before bidding or by entering in-house models that act as reserve prices.\looseness=-1
    \item \textbf{Empirical analysis (\S\ref{sec:experiments}):} On EmbedLLM~\citep{zhuang2025embedllm}, we find that competitive floors and conditional independence emerge and yield high efficiency, even in adversarial equilibria.\looseness=-1
\end{enumerate}

\section{Preliminaries}
\label{sec:preliminaries}

We begin by briefly reviewing the first-price auction (FPA).  Consider a single user prompt that is to be serviced by one of $n$ bidders (also called ``providers'' or ``agents'') for some user-provided reward.\footnote{Where do these rewards come from? The theory that follows is agnostic to this question, but a few concrete  instantiations can be  helpful to keep in mind. The simplest is a binary reward such as a thumbs up/down reported back by the platform or user (assumed, honestly). More generally, the platform could use any agreed-upon function $\phi(q, y)$ that takes a prompt $q$ and response $y$ and outputs a reward $r$ funded by the user (e.g., a reward in escrow is released contingent upon the agent's response $y$ passing a collection of unit tests or clearing a third-party autorater / verifier).\looseness=-1} Each bidder $i$ has a value $v_i$ for the request, defined as its expected reward minus its cost of service:
\begin{equation}
v_i \;=\; \E[\mathrm{reward}(\text{response}_i) - \text{cost}_i \mid \text{query}].
\end{equation}
In a standard auction, the item for sale is the only source of randomness: each bidder knows its own value, which is random only from the perspective of others. But in agentic markets, the reward and cost of a response is typically not known until it is produced. Simply letting bidder $i$ take the expectation of the response reward minus cost reduces the setting to the standard one. If bidder $i$ wins the request and pays a price $p_i$, its expected payoff (or \emph{surplus}) is then simply $v_i - p_i$.\looseness=-1

In a sealed-bid FPA, every bidder $i$ submits a bid $b_i$. The bidder with the highest bid wins, and pays its own bid. We write $i^*(\mathbf{b})$ for the index of the winning bidder where $\mathbf{b} = (b_1, \dots, b_n)$, with ties broken by an announced measurable bid-only rule, possibly randomized. A rational bidder, however, may prefer not to bid its true value, as doing so guarantees zero expected surplus. Bidders may therefore strategically \emph{shade} their bids below their values to gain positive surplus. We model this strategic interaction as a Bayesian game of incomplete information~\cite{RoughgardenPOA}. Specifically, we model bidder values $\mathbf{v} = (v_1, \ldots, v_n)$ as draws from a joint value prior $\mu$, and a strategy for bidder $i$ is a possibly randomized function $s_i$ mapping a value $v_i$ to a bid $b_i = s_i(v_i)$. A strategy profile $\mathbf{s} = (s_1, \dots, s_n)$ is a \emph{Bayes-Nash equilibrium} (BNE) if, for every bidder $i$ and almost every value $v_i$, the bid $s_i(v_i)$ maximizes bidder $i$'s expected surplus given the strategies of the other bidders ($\mathbf{s}_{-i}$) and the value prior ($\mu$).\looseness=-1

On a technical level, we assume throughout that bidders never bid more than their value at equilibrium, i.e., $0\le s_i(v_i) \le v_i$ almost surely for each $i \in [n]$ and almost every $v_i\ge0$. This is standard in the  literature, since overbidding risks negative surplus \citep{SyrgkanisTardos2013,JinLu2022}.\looseness=-1

Throughout the rest of this paper, expectations and probabilities are taken over the joint randomness in values, private bidding randomization, and tie-breaking. We write $\E_\mu$ and $\Prb_\mu$ only when more than one value distribution is in play, as in the PoA~\eqref{eq:poa}; when the prior is fixed, we write $\E$ and $\Prb$.\looseness=-1

\subsection{Social welfare and the price of anarchy}

It is natural to measure an auction's performance by the total \emph{social welfare} (which we will just refer to as the \emph{welfare}) it generates, which is the sum of the bidder utilities, plus the revenue of the auctioneer.\footnote{We model a user's value as being equal to the price they pay for the response, so that their utility is always $0$.} In the simple single-item setting we consider for agentic markets, the welfare reduces to the value of the winning bidder, $W(\mathbf{b}, \mathbf{v}) := v_{i^*(\mathbf{b})}$. The \emph{optimal welfare} is then $\OPT(\mathbf{v}) := \max_{i} v_i$.\footnote{We do not consider externalities, such as $i^*$ getting additional private utility by placing an advertisement.}

Since FPAs allocate to the highest \emph{bidder} rather than the highest-value provider, shading can cause the auction to select a suboptimal provider. The \emph{price of anarchy} (PoA) quantifies the worst-case welfare loss due to strategic behavior. Formally, the PoA is the infimum, over all settings (number of bidders, joint value distribution $\mu$) and BNE, of the ratio of the expected welfare to the optimal welfare, i.e.,\looseness=-1
\begin{equation}
\label{eq:poa}
\text{PoA} := \inf_\mu \inf_{\mathbf{s} \in \mathrm{BNE}(\mu)} \frac{\E_{\mu}[W(\mathbf{s}(\mathbf{V}), \mathbf{V})]}{\E_{\mu}[\OPT(\mathbf{V})]}.
\end{equation}
The infimum is restricted to priors with $0<\E_\mu[\OPT(\mathbf V)]<\infty$ and a nonempty set of equilibria (that, recall, we assume are non-overbidding). In symmetric settings (where all bidders draw values from the same distribution), the BNE of a first-price auction is unique and fully efficient~\cite{chawla2013unique}, and the PoA is $1$. In asymmetric settings, which arise when bidders have different expected values, as expected in agentic markets, equilibria are generally inefficient.\looseness=-1

A foundational result establishes that every BNE of a first-price single-item auction achieves $\geq (1-1/e) \approx 63.2\%$ of optimal welfare~\cite{SyrgkanisTardos2013, RoughgardenPOA}. 
Under the stronger assumption that bidders' values are drawn independently (but not identically), \citet{JinLu2022} obtain the tight bound of $(1 - 1/e^2) \approx 86.5\%$. These bounds are important, though they do paint a pessimistic picture of first-price auctions. A potential welfare loss of $37\%$ (or even $13.5\%$) is large enough to make a platform hesitant to deploy first-price orchestration at all. These bounds, however, are worst-case lower bounds over all possible value distributions. As we will see, agentic markets often possess structural features that push their realized welfare higher.\looseness=-1
\section{First-price auctions with competitive floors}
\label{sec:concentration}

The price of anarchy fundamentally depends on what is at stake. If all bidders value an item identically, then any allocation achieves full efficiency. More generally, no matter which bidder wins, welfare is at least the lowest bidder value, which therefore acts as a "floor" on the welfare. Write $V_{(1)}\ge\cdots\ge V_{(n)}$ for the bidder values in decreasing order. Then for any given value prior $\mu$, we must have\looseness=-1
\begin{equation}
\label{eq:min_max_bound}
\inf_{\mathbf{s} \in \mathrm{BNE}(\mu)} \frac{\mathbb{E}[W(\mathbf{s}(\mathbf{V}), \mathbf{V})]}{\mathbb{E}[\OPT(\mathbf{V})]}\ge\frac{\mathbb{E}[V_{(n)}]}{\mathbb{E}[V_{(1)}]},
\end{equation}
which does not require strategic analysis. Note that  \Cref{eq:min_max_bound}, as well as the other bounds we will present here, are specific to the prior $\mu$. \Cref{sec:priors} then shows how these bounds are uniform over any family of priors that satisfy the same minimal conditions. Looking at \Cref{eq:min_max_bound}, when  $\E[V_{(n)}]$ is close to $\E[V_{(1)}]$ (e.g., when all models perform similarly on user queries), the bound is close to $1$. A single weak bidder, however, might pull it far below $1$. The rest of this section analyzes competition in the auction to obtain a bound that depends only on the strongest few bidders.\looseness=-1

Intuitively, markets only require competition at the top (that is, just between the highest-valued bidders) to be efficient. In other words, it should be enough for more than one, but not all, bidders to usually value the query highly. We formalize this via a notion of a "competitive floor" for the market.\looseness=-1

\begin{definition}[Competitive floor]
\label{def:competitive_floor}
Fix a fraction $\alpha \in [0, 1]$ of the optimal welfare, a number of bidders $m \in \{2, \ldots, n\}$, and a failure rate $\epsilon \in [0, 1]$. A market has an \emph{$(\alpha, m, \epsilon)$-competitive floor} if at least $m$ bidders value the query at $\alpha \E[V_{(1)}]$ or higher with probability at least $1-\epsilon$, i.e., $\mathbb{P}\bigl(V_{(m)} \ge \alpha \E[V_{(1)}]\bigr) \ge 1-\epsilon$.
\end{definition}

\subsection{Deterministic competitive floors}
\label{sec:deterministic_floor}

As a warmup, suppose two bidders \emph{always} clear $\alpha \E[V_{(1)}]$. The PoA bound then rises with $\alpha$.\looseness=-1

\begin{proposition}
\label[proposition]{proposition:deterministic_poa}
Suppose $0<\mathbb{E}[V_{(1)}]<\infty$ and $V_{(2)} \ge \alpha \mathbb{E}[V_{(1)}]$ almost surely for $\alpha \in [0, 1]$. Then\looseness=-1
\begin{equation}
\inf_{\mathbf{s} \in \mathrm{BNE}(\mu)} \frac{\mathbb{E}[W(\mathbf{s}(\mathbf{V}), \mathbf{V})]}{\mathbb{E}[\OPT(\mathbf{V})]}
 \ge \left(1 - \frac{1}{e}\right) + \frac{\alpha}{e}.
\label{eq:deterministic_poa_general}
\end{equation}
\end{proposition}

All proofs are in the appendix. At a high level, we first prove that deterministic competition from the runner-up ensures that the equilibrium winning bid is at least $\alpha\E[V_{(1)}]$, which also guarantees that the welfare is at least $\alpha\E[V_{(1)}]$. Any equilibrium then  further recovers at least a $(1-1/e)$-fraction of the optimal welfare remaining \emph{above} the floor, i.e., $(1-\alpha)\E[V_{(1)}]$. As an example, suppose a rival always values the query at least half the expected top value ($\alpha = 0.50$). The bound rises from $\approx63\%$ to $\approx82\%$.\looseness=-1

\subsection{Probabilistic competitive floors}
\label{sec:probabilistic_floor}

\Cref{proposition:deterministic_poa} corresponds to a competitive floor with $m = 2$ and $\epsilon = 0$. That is, it requires at least two bidders to clear the floor on \emph{every} query. This is not realistic: on unusual queries for instance, only one model may be able to do well, even if several models usually compete closely on most queries. When bidders know that a second credible competitor is sometimes lacking, they can shade more, and the winning bid may fall below the floor. Combining a bound on how far it can fall (\Cref{app:bid_shortfall}) with a similar argument to that of \Cref{proposition:deterministic_poa} gives the following bound.\looseness=-1


\begin{figure}[t]
\centering

\pgfplotsset{
  every axis/.append style={
    width=0.35\textwidth,
    height=0.23\textwidth,
    xmode=log,
    log basis x=2,
    xtick={2,4,8,16,32,64},
    xticklabels={\llap{$m{=}\,$}2,4,8,16,32,64},
    xmin=1.6, xmax=80,
    ymin=0.60, ymax=1.0,
    ylabel near ticks,
    grid=major,
    grid style={gray!30},
    tick label style={font=\tiny},
    label style={font=\scriptsize},
    title style={font=\scriptsize, yshift=-6pt},
    cycle list name=exotic,
  },
}

\pgfplotsset{
  eps00/.style={color=blue!25!white,  line width=1pt, mark=*, mark size=2pt},
  eps01/.style={color=blue!45!white,  line width=1pt, mark=square*, mark size=2pt},
  eps05/.style={color=blue!60!black,  line width=1pt, mark=triangle*, mark size=2pt},
  eps10/.style={color=blue!35!black,  line width=1pt, mark=diamond*, mark size=2pt},
  eps15/.style={color=blue!15!black,  line width=1pt, mark=pentagon*, mark size=2pt},
  floorline/.style={color=red!75!black, densely dashed, line width=0.9pt, no marks},
}

\begin{tikzpicture}
\begin{groupplot}[
  group style={
    group size=3 by 1,
    horizontal sep=0.45cm,
    ylabels at=edge left,
    yticklabels at=edge left,
  },
  ylabel={PoA bound},
]

\nextgroupplot[title={$\alpha = 0.50$}]
\addplot[eps00] coordinates {(2,0.8161) (4,0.8161) (8,0.8161) (16,0.8161) (32,0.8161) (64,0.8161)};
  \label{plot:eps00}
\addplot[eps01] coordinates {(2,0.7455) (4,0.7937) (8,0.8045) (16,0.8081) (32,0.8096) (64,0.8103)};
  \label{plot:eps01}
\addplot[eps05] coordinates {(2,0.6599) (4,0.7420) (8,0.7692) (16,0.7803) (32,0.7855) (64,0.7881)};
  \label{plot:eps05}
\addplot[eps10] coordinates {(2,0.6321) (4,0.6927) (8,0.7306) (16,0.7478) (32,0.7564) (64,0.7608)};
  \label{plot:eps10}
\addplot[eps15] coordinates {(2,0.6321) (4,0.6506) (8,0.6949) (16,0.7166) (32,0.7278) (64,0.7338)};
  \label{plot:eps15}
\addplot[floorline] coordinates {(1.6, 0.63212) (80, 0.63212)};
  \label{plot:floor}

\nextgroupplot[title={$\alpha = 0.70$}]
\addplot[eps00] coordinates {(2,0.8896) (4,0.8896) (8,0.8896) (16,0.8896) (32,0.8896) (64,0.8896)};
\addplot[eps01] coordinates {(2,0.7909) (4,0.8584) (8,0.8735) (16,0.8785) (32,0.8806) (64,0.8816)};
\addplot[eps05] coordinates {(2,0.6711) (4,0.7860) (8,0.8240) (16,0.8396) (32,0.8469) (64,0.8505)};
\addplot[eps10] coordinates {(2,0.6321) (4,0.7170) (8,0.7700) (16,0.7941) (32,0.8061) (64,0.8123)};
\addplot[eps15] coordinates {(2,0.6321) (4,0.6580) (8,0.7200) (16,0.7504) (32,0.7661) (64,0.7745)};
\addplot[floorline] coordinates {(1.6, 0.63212) (80, 0.63212)};

\nextgroupplot[title={$\alpha = 0.90$}]
\addplot[eps00] coordinates {(2,0.9632) (4,0.9632) (8,0.9632) (16,0.9632) (32,0.9632) (64,0.9632)};
\addplot[eps01] coordinates {(2,0.8363) (4,0.9230) (8,0.9425) (16,0.9489) (32,0.9516) (64,0.9529)};
\addplot[eps05] coordinates {(2,0.6822) (4,0.8299) (8,0.8788) (16,0.8989) (32,0.9082) (64,0.9129)};
\addplot[eps10] coordinates {(2,0.6321) (4,0.7412) (8,0.8093) (16,0.8404) (32,0.8558) (64,0.8638)};
\addplot[eps15] coordinates {(2,0.6321) (4,0.6655) (8,0.7452) (16,0.7842) (32,0.8044) (64,0.8152)};
\addplot[floorline] coordinates {(1.6, 0.63212) (80, 0.63212)};

\end{groupplot}

\node[anchor=west, font=\scriptsize, draw=gray!60, rounded corners=2pt,
      inner xsep=3pt, inner ysep=2pt, fill=white, align=left]
      at ([xshift=0.25cm]group c3r1.east) {%
  \ref{plot:eps00}\,$\epsilon{=}0$\\
  \ref{plot:eps01}\,$\epsilon{=}0.01$\\
  \ref{plot:eps05}\,$\epsilon{=}0.05$\\
  \ref{plot:eps10}\,$\epsilon{=}0.10$\\
  \ref{plot:eps15}\,$\epsilon{=}0.15$\\
  \ref{plot:floor}\,$1{-}1/e$%
};

\end{tikzpicture}

\caption{Price-of-anarchy bounds under competitive floors, for various settings of $(\alpha, m, \epsilon)$.
  Each entry represents the bound for any distribution with $\mathbb{P}(V_{(m)} < \alpha \mathbb{E}[V_{(1)}]) \leq \epsilon$.  At $\epsilon = 0$, we recover the deterministic floor bound from \Cref{eq:deterministic_poa_general}; higher $\epsilon$ start to degrade the efficiency. Larger numbers of credible competitors ($m \ge 4$) stabilize the market, and preserve efficiency even under high $\epsilon$.\looseness=-1
  }
\label{fig:poa_correlated_value}
\vspace{-10pt}
\end{figure}
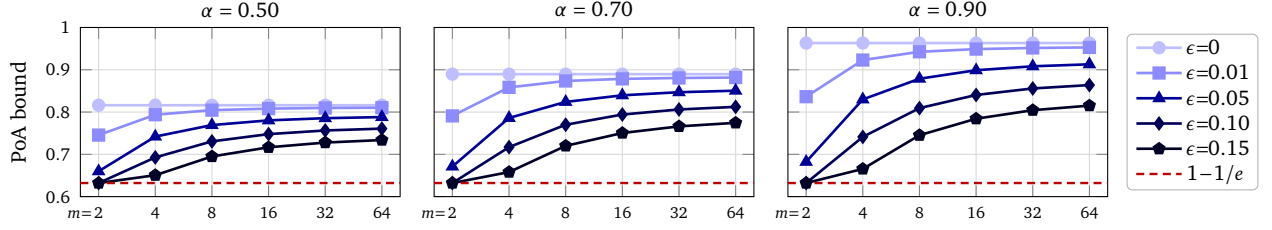

\begin{theorem}
\label{thm:multi_floor}
Suppose $0 < \mathbb{E}[V_{(1)}] < \infty$ and the market has an $(\alpha, m, \epsilon)$-competitive floor. Then
\begin{equation}
\inf_{\mathbf{s} \in \mathrm{BNE}(\mu)} \frac{\mathbb{E}[W(\mathbf{s}(\mathbf{V}), \mathbf{V})]}{\mathbb{E}[\OPT(\mathbf{V})]}
 \;\ge\; \left(1 - \frac{1}{e}\right) + \alpha\left[\frac{1}{e} - S_m(\epsilon)\right]_+,
\label{eq:multi_floor_bound}
\end{equation}
where $S_m(\epsilon) := I_\epsilon(m-1, 2)^{1/m}$, $I_\epsilon$ is the regularized incomplete beta function, and $[\;\cdot\;]_+ = \max(0, \;\cdot\;)$.
\end{theorem}
\begin{remark}
\label{remark:rate}
$S_m(\epsilon)$ is a penalty that increases from $0$ to $1$ as $\epsilon$ goes from $0$ to $1$, and decreases toward $\epsilon$ as $m$ grows. It also satisfies the simple bound $S_m(\epsilon)\le m^{1/m}\epsilon^{1-1/m}$. For example, it is at most $\sqrt{2\epsilon}$ when $m=2$.\looseness=-1
\end{remark}

Compared to \Cref{proposition:deterministic_poa}, the bound in \Cref{thm:multi_floor} subtracts a penalty $\alpha S_m(\epsilon)$ that depends, in a sense, on how reliably the competition holds. The penalty is zero when $\epsilon=0$, in which case \Cref{thm:multi_floor} recovers \Cref{proposition:deterministic_poa}, and it shrinks when more bidders ($m$) clear the floor (at the same rate, $ 1- \epsilon$).\looseness=-1

To further illustrate how the bound in \Cref{eq:multi_floor_bound} behaves, in \Cref{fig:poa_correlated_value} we plot its value while varying $m$, $\alpha$, and a hypothetical failure rate $\epsilon$.
We see that even if the failure rate $\epsilon$ is high (e.g., $\approx 0.10$), we can still guarantee a stronger PoA if the floor holds for several competitors. For example, suppose $m$ models achieve $90\%$ of the mean best value at least $90\%$ of the time. If only two models clear this threshold ($m=2$), the bound stays at $\approx 63\%$. With $m=4$ it rises to $\approx74\%$, and with $m=16$ to $\approx 84\%$.\looseness=-1 

For a given $\mu$, \Cref{thm:multi_floor} holds for all $\alpha \in [0, 1]$ and $m \in \{2, \ldots, n\}$, so we can take the supremum over both.\footnote{Likewise, since the non-strategic bound in \Cref{eq:min_max_bound} holds for the same $\mu$, the larger of the two bounds also holds.} For each $(\alpha, m)$, the smallest valid failure rate is $\epsilon^*_{\alpha, m}= \mathbb{P}_\mu(V_{(m)} < \alpha \mathbb{E}_\mu[V_{(1)}])$, which gives\looseness=-1
\begin{equation}
\inf_{\mathbf{s} \in \operatorname{BNE}(\mu)} \frac{\mathbb{E}[W(\mathbf{s}(\mathbf{V}), \mathbf{V})]}{\mathbb{E}[\mathrm{OPT}(\mathbf{V})]} \;\ge\; \Bigl(1 - \frac{1}{e}\Bigr) + \sup_{\alpha, m} \alpha\Bigl[\frac{1}{e} - S_m(\epsilon^*_{\alpha, m})\Bigr]_+.
\label{eq:sup-bound}
\end{equation}

\Cref{sec:additional_theory} contains additional results on competitive floors. \Cref{sec:conditional_welfare_bounds} conditions \Cref{thm:multi_floor} on public information $X$ about the query, such as the prompt text. The floor $\alpha\,\E[V_{(1)}\mid X=x]$ then adapts to how hard each query is, and the resulting bound is never weaker than the pooled one. \Cref{sec:positive_dependence} also explores  a related setting in which positively dependent (e.g., affiliated) values always raise the bound in \Cref{eq:min_max_bound} compared to independent values with the same marginal bidder distributions.\looseness=-1

\section{First-price auctions with near-independent values}
\label{sec:near_independence}

In the previous section we showed that competitive floors on the values limit welfare loss. A separate question is how the \emph{dependence structure}  among values (or lack thereof) affects strategic behavior. More specifically, under independent private values (IPV; where $V_1 \perp \dots \perp V_n$), \citet{JinLu2022} show that the PoA improves from $1 - 1/e \approx 0.632$ to $1 - 1/e^2 \approx 0.865$. In practice, bidder values are rarely \emph{strictly} independent---yet they may be nearly independent. In agentic markets, values co-vary across queries due to shared task difficulty and similar model costs, but we conjecture that much of this common variation can be explained by observable query metadata. If so, once query metadata $X$ is observed (e.g., the prompt text, its topic, how response quality is measured, the model tiers of the other bidders, etc.), agent values may be approximately independent, where $\mu(v_1, \ldots, v_n \mid X) \;\approx\; \prod_{i=1}^n \mu_i(v_i \mid X).$\looseness=-1

Whether this holds in practice is an empirical question that we examine in \Cref{sec:approx_condt_indep}. But first we must ask whether we can approach the $1-1/e^2$ guarantee with \emph{approximate} conditional independence. We prove that we can. Specifically, for each realization of the public state $X = x$, let $\mu_x := \mu(\cdot \mid X = x)$ and $\nu_x := \prod_{i=1}^n \mu_i(\cdot \mid X = x)$, and write $ \tau(x) \;:=\; \|\mu_x - \nu_x\|_{\mathrm{TV}}$ for the total variation distance between the conditional joint distribution and the product of its marginals. We use $\|\mu-\nu\|_{\mathrm{TV}}:=\sup_A|\mu(A)-\nu(A)|$. When the conditional total variation distance $\tau(x) = 0$, values are conditionally independent given $X = x$. We then can show the following qualitative result.\looseness=-1

\begin{theorem}
\label{thm:near_independence_welfare}
For every fixed $n \ge 2$, there is a nondecreasing function $\eta_n^{\mathrm{cont}} \colon [0,1] \to [0,\infty)$ with $\eta_n^{\mathrm{cont}}(0) = 0$ and $\eta_n^{\mathrm{cont}}(t) \to 0$ as $t \downarrow 0$ such that the following holds. Suppose $0 \le V_i \le \overline{v}(X) < \infty$ almost surely. Then every $\mathbf{s}\in \mathrm{BNE}(\mu)$ satisfies, for almost every $x$,
\begin{equation}
\E[W(\mathbf{s}(\mathbf{V}),\mathbf{V}) \mid X = x]
\;\ge\;
\left(1 - \frac{1}{e^2}\right) \E[V_{(1)} \mid X = x]
\;-\; \overline{v}(x)\,\eta_n^{\mathrm{cont}}(\tau(x)).
\label{eq:continuum_near_independence_level_bound}
\end{equation}
\end{theorem}
When $\tau(x) = 0$, the error term vanishes and we recover the $1 - 1/e^2$ bound exactly, even if values are arbitrarily correlated \emph{across} queries. Small departures from conditional independence incur an additive penalty scaled by the conditional value bound $\overline{v}(x)$. Converting to a ratio bound requires that $\overline{v}(x)$ is not too large relative to $\E[V_{(1)} \mid X = x]$; integrating over $X$ gives the unconditional guarantee. The proof accounts for changes in both welfare expectations and bidding incentives when the joint law is replaced by the product law. Appendix~\ref{sec:independence_proofs} gives the full argument, with additional refinements.\looseness=-1

Note that competitive floors and near-independence are not mutually exclusive, despite perhaps seeming so. For instance, suppose $V_i=a(X)+\theta Z_i$, where $a(X)\ge\theta$ and, conditional on $X$, $Z_i$ are zero-mean, independent random variables in $[-1,1]$. Then $V_i$ are conditionally independent and lie in an interval of width $2\theta$. When $\theta$ is small relative to $a(X)$, they also share a high competitive floor.\looseness=-1

\section{On implications for agentic market platform design}
\label{sec:applications}

The market structure supporting our bounds is partly under the platform's control, leading to some very simple design recommendations to protect market efficiency. In this section we explore what happens if the platform restricts which models bid on a query, or if it honestly enters its own in-house model (e.g., if it has privately developed one, or hosts an open-weight model) into the auction.

\textbf{Pre-auction filtering.}
Many platforms already run a lightweight router that picks a subset $\mathcal{S}(X) \subseteq [n]$ of models based on the public query information $X$~\cite[\textit{inter alia}]{hu2024routerbench, ong2025routellm}. Let $W_{\mathcal{S}}$ be the welfare at an equilibrium of the auction in which only models in $\mathcal{S}(X)$ bid (note that the randomness in $W_\mathcal{S}$ also depends on $\mathcal{S}(X)$). The winner of this auction is always in $\mathcal{S}(X)$ by definition, so simply applying the naive bound of \Cref{eq:min_max_bound} within it, and marginalizing over $X$, gives\looseness=-1
\begin{equation}
 \frac{\E[\min_{i \in \mathcal{S}(X)} V_i]}{\E[V_{(1)}]} \;\le\; \frac{\E[W_{\mathcal{S}}]}{\E[V_{(1)}]} \;\le\; \frac{\E[\max_{i \in \mathcal{S}(X)} V_i]}{\E[V_{(1)}]}.
\label{eq:router_welfare_bound}
\end{equation}
While this bound might be weak for a market with open participation (where a single poor model can drag the welfare bound down, as discussed at the start of \Cref{sec:concentration}), a platform that \emph{controls} who competes can make it stronger by removing weak models, which raises the lower bound. Removing the best model (i.e., if $V_{(1)} \not\in \mathcal{S}(X))$, however, also lowers the upper bound; thus filtering does incur some risk. Still, that risk can be mitigated by being more conservative in filtering, and as long as the router keeps the strongest models, the bounds of \Cref{sec:conditional_welfare_bounds} still apply to the filtered market $\mathcal{S}(X)$.\looseness=-1

\textbf{Platform participation.}
A platform can also enter its own in-house model into the auction. We study a simple version of this, in which the platform publicly announces its model's value $V_0$ for each query and then bids exactly $V_0$. Since $V_0$ is then part of the public state $X$, every model can condition on it, and the platform's model acts as a query-specific reserve price. This mirrors private-label strategies from retail economics~\cite[see, e.g.,][]{SHOPOVA2023102949, ding2025marketplaceoperatorsinducecompetitive}. Similar to the deterministic competitive floor in \Cref{sec:deterministic_floor}, the platform's bid keeps the winning bid, and hence the welfare, at least $V_0$. On top of this, any equilibrium recovers at least a $(1-1/e)$-fraction of the optimal welfare above $V_0$. Write $\mathbf{V}=(V_0,V_1,\dots,V_n)$ and $\mathbf{s}=(s_0,s_1,\dots,s_n)$ with $s_0(V_0)=V_0$, so that $V_{(1)}=\max_{i\ge0}V_i$ includes the platform's model. Then every $\mathbf{s}\in\mathrm{BNE}(\mu)$ satisfies (see \Cref{app:platform_reserve} for the derivation):\looseness=-1
\begin{equation}
\frac{\E[W(\mathbf{s}(\mathbf{V}),\mathbf{V})]}{\E[V_{(1)}]} \;\ge\; \left(1 - \frac{1}{e}\right) + \frac{1}{e}\,\frac{\E[V_0]}{\E[V_{(1)}]}.
\label{eq:platform_reserve_poa}
\end{equation}
Notice that if $\E[V_0] \ge \alpha\E[V_{(1)}]$, this bound is at least that of \Cref{proposition:deterministic_poa}. The assumptions of \Cref{proposition:deterministic_poa}, however, require a rival $V_{(2)} \geq \alpha\E[V_{(1)}]$ almost surely (i.e., on \emph{every} query), whereas the platform model's value (which is effectively a rival) only needs to exceed $\alpha\E[V_{(1)}]$ on average.
\section{Experiments}
\label{sec:experiments}

We now test how well our bounds describe welfare in a realistic LLM market. We run one auction with all models on the full EmbedLLM dataset~\cite{zhuang2025embedllm}, compare the competitive-floor bound to realized welfare over many random sub-markets, and then also study hand-picked, illustrative cohorts that fix optimal welfare while varying market structure. The appendix contains additional results that analyze, empirically, how weakly dependent values are after conditioning on the query (\Cref{sec:approx_condt_indep}), as well as results with platform-owned bidders and pre-auction filtering (\Cref{app:platform_bidder_exp,app:router_exp}).\looseness=-1

\subsection{Setting}
\label{sec:setting}
\textbf{Dataset.} We simulate auctions with strategic bidders on the EmbedLLM routing dataset~\citep{zhuang2025embedllm}, which aggregates over $30k$ prompts across math, code, question answering, and reasoning benchmarks, along with the scored responses from $112$ models of various sizes. To align with our convention that a model's value is its {expected reward} minus its cost, we set $v_i(q) = \max(0, \hat{\E}[\mathrm{acc}_i \mid q] - c_i)$, where $\hat{\E}[\mathrm{acc}_i \mid q]$ estimates the probability that model $i$ answers query $q$ correctly, computed by $k$-nearest neighbors ($k = 16$) over {Gemini} prompt embeddings~\citep{shanbhogue2026geminiembedding2native}, excluding $q$. Following \citet{jitkrittum2025universal}, we set $c_i$ to be linear in parameter count, at $1$ per trillion parameters, so that, e.g., {Qwen-1.5-7B-Chat} has a cost of $.00772$ and {Qwen-1.5-32B-Chat} has a cost of $.0325$. Two models were excluded because their parameter count could not be determined, leaving $110$ total.\looseness=-1

\textbf{Strategic bidding.} Each model submits a sealed bid $b_i(q) \le v_i(q)$ for each query $q$, and the highest bidder wins and pays its bid. We partition the prompts into $20$ $K$-means clusters using their embeddings, and bidders condition their strategies on the cluster and their own value $v_i(q)$. A bidder therefore knows the approximate type of prompt it is bidding on and its own estimated value, but knows its rivals' values only through the prior distribution, conditioned on the cluster. Within each cluster we solve for (approximate) Bayes--Nash equilibria by fictitious play using a discretized grid over both values and possible bids (see \Cref{app:solver} for details). Equilibria are not necessarily unique, and certain equilibria might have substantially worse welfare than others. We therefore solve every market twice: once for a \emph{typical} equilibrium, where at each iteration bidders update their strategy to the best (surplus maximizing) response nearest its current bid until convergence, and once for an (also approximate) \emph{adversarial} equilibrium, where at each iteration each bidder takes the welfare-\emph{minimizing} response that still retains at most $5\%$ relative regret of its best possible individual surplus until convergence.\looseness=-1

\textbf{Metrics.} Welfare within a market is pooled over all prompts and reported as a fraction of the first-best oracle, $\OPT = \E[\max_i v_i(q)]$, which we call the ``realized PoA'' to distinguish it from the worst-case bounds of \Cref{sec:concentration}.  By a ``market'' we mean a subset of the $110$ models bidding against each other, and most of our experiments draw many such subsets at random, with the sampling scheme stated in each case. Every number we report is therefore computed over prompts and over sampled markets. Since the bounds we are testing are worst-case statements, we report both the mean and the $0.1$-percentile (\texttt{p0.1}) of realized PoA across markets,\footnote{We use $99.9\%$ coverage instead of the minimum to avoid outlier markets where the solver has not converged.} and we compare both against the competitive-floor bound of \Cref{thm:multi_floor} evaluated on the same market (where like the equilibrium strategies, the bound is computed using the empirical distribution of the bidder values). We also track how much the winner shades, $1 - \E[B_\mathrm{winner}/V_\mathrm{winner}]$. Note, however, that shading (i.e., low winning bids) on its own does not cost welfare, since queries are misallocated only when bidders shade by (problematically) different amounts. Finally, as bidders condition on the cluster (see "Strategic bidding" above), each of the $20$ clusters is effectively its own auction, and the pooled welfare averages over them. As a robustness measure, we also report welfare on the \emph{worst cluster}, i.e., the cluster where a given allocation rule does worst.\looseness=-1

\subsection{Results}

\begin{table}[!t]
\centering
\small

\setlength{\tabcolsep}{4.5pt}
\begin{tabular}{llcccc|cc}
\toprule
& & Oracle & Random  & Best single & Best per cluster & FPA & FPA (adv) \\
\midrule
\multirow{2}{*}{Full dataset}       & Empirical welfare       & $0.768$ & $0.416$ & $0.600$ & $0.686$ & $0.762$ & $0.758$ \\
                               & Empirical PoA & $1.000$ & $0.542$ & $0.781$ & $0.894$ & $0.992$ & $0.988$ \\
\midrule
\multirow{1}{*}{Worst cluster} 
                               & Empirical PoA & $1.000$ & $0.083$ & $0.060$ & $0.640$ & $0.960$ & $0.930$ \\
\bottomrule
\end{tabular}
\caption{Realized welfare on the full EmbedLLM market under different allocation rules, raw and as a fraction of the optimal welfare. As baselines, \emph{Random} allocates queries randomly, while \emph{Best single} and \emph{Best per cluster} use the model with the highest average value overall or within each cluster, respectively; the latter is a router that sees only the prompt's cluster, the same public state the bidders condition on. \emph{FPA} and \emph{FPA (adv)} use a typical and an adversarially chosen (approximate) equilibrium (\S\ref{sec:setting}). Worst-cluster welfare is averaged on each rule's lowest-welfare  prompt cluster (out of $20$).\looseness=-1}
\vspace{-10pt}
\label{tab:rules}
\end{table}

\begin{figure}[t]
\centering
\input{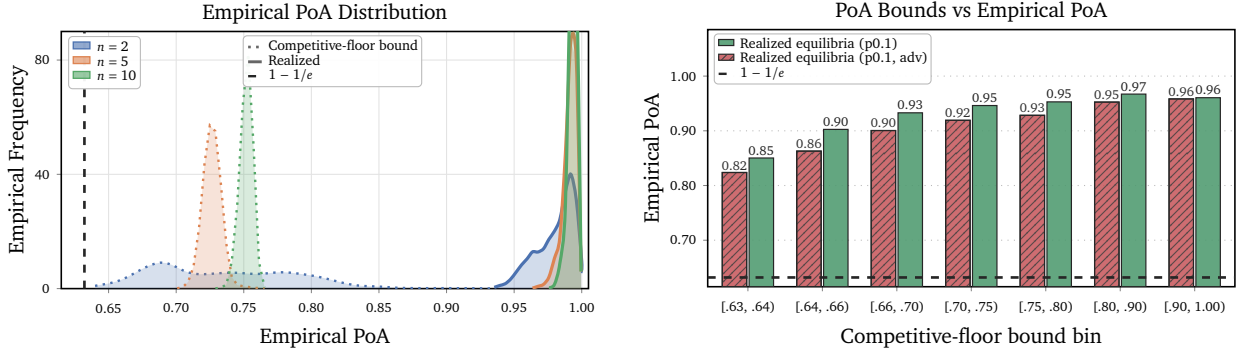}
\vspace{-5pt}
\caption{\textbf{Left:} Empirical distributions of the \Cref{thm:multi_floor} bound (dotted) and the empirical PoA (solid) across $100k$ random $n$-model subsets of EmbedLLM, for $n \in \{2, 5, 10\}$. Both distributions tighten as $n$ grows, and are well above the distribution-free baseline (dashed). \textbf{Right:}~Bottom $0.1\%$ of empirical PoA across sub-markets grouped by bound bin ($100k$ samples per bin), for both typical equilibria (solid green) and approximate, adversarial equilibria (hatched red).\looseness=-1}
\vspace{-10pt}
\label{fig:bound_buckets}
\end{figure}

\textbf{Results on the full EmbedLLM market.} \Cref{tab:rules} reports a single auction run on the full dataset with all $110$ models bidding against each other. The auction at equilibrium is nearly perfectly efficient, reaching $99.2\%$ of oracle welfare at the typical equilibrium and $98.8\%$ at the adversarial one. Both exceed the baseline of routing each cluster to its own best model ($89.4\%$), and are well above the baseline of committing to one model everywhere ($78.1\%$). The gap widens on each rule's worst cluster, where the single best model retains only $6.0\%$ of first-best and per-cluster routing $64.0\%$, while the auction holds $96.0\%$ and $93.0\%$. Bidders see their own value $v_i(q)$ before bidding, so a specialist can outbid the cluster leader on exactly the prompts where the leader is weak, which is what a fixed rule cannot do. On the full market and prompt distribution, \Cref{thm:multi_floor} guarantees a welfare ratio $\geq 79.9\%$ (and a minimum bound of $67.6\%$ on the worst cluster). Both are substantially above the naive $1 - 1/e$ bound, though still below what the market achieves. The $1 - 1/e$ bound holds for every market and so cannot indicate anything more specific about this one, whereas $79.9\%$ is computed from the empirical value prior of this market, and correctly signals that welfare will be high. Note that we do not expect it to be tight, since the bound only uses the competitive floor and treats everything else as worst case, while near-independence (\Cref{sec:near_independence}) is a second and separate reason this market can be efficient.\looseness=-1

\textbf{Results on randomly sampled EmbedLLM sub-markets.} Moving beyond the full EmbedLLM market setting, we next look at welfare in thousands of auctions among subsets of EmbedLLM models. Figure~\ref{fig:bound_buckets} (left) plots the empirical distributions of both (a) the competitive-floor bound of \Cref{thm:multi_floor}, and (b) the realized welfare ratio at a typical equilibrium over $100k$ randomly sampled model subsets of market size $n \in \{2, 5, 10\}$. Both distributions concentrate as $n$ grows, and their lower tails move up. Realized welfare is much higher than $1 - 1/e$, as are our bounds (though while the competitive-floor bounds are tighter, they are still conservative). In practice, the realized welfare ratio of the vast majority of sampled auctions falls above $95\%$. That said, unlike the worst-case $1 - 1/e$ bound, which is the same for every market, the competitive-floor bounds are \emph{sensitive} to the market they are computed for. \Cref{fig:bound_buckets} (right) groups randomly sampled model subsets for market sizes $n \in \{2, \ldots, 25\}$ into bins according to their competitive-floor bound, where the bins span $[1 - {1}/{e}, 1.00)$. We sample $100k$ subsets per bin, and report the bottom $0.1\%$ of realized welfare ratios within each bin, for both the typical and the adversarial strategy solver. The realized welfare at an adversarial equilibrium rises with the bins, from $0.82$ in the $[0.63, 0.64)$ bin to $0.96$ in the $[0.90, 1.00)$ bin. At a typical equilibrium it rises from $0.85$ in the lowest bin to $0.97$ in the $[0.80, 0.90)$ bin, and is $0.96$ in the top bin.\looseness=-1

\textbf{Results on illustrative EmbedLLM cohorts.} Random sub-markets vary in both how good their best model is and how much competition it faces. To separate the two, \Cref{tab:cohorts} shows results for four hand-picked cohorts of $n=4$ models with nearly identical first-best welfare ($\E[\max_i V_i] \approx 0.55$) but different market structure. In the \emph{similar} cohort, the four models are comparable $7$B chat models ({Starling-LM-7B-beta}, {OpenChat-3.5-0106}, {Zephyr-7B-beta}, and {Mistral-7B-Instruct-v0.1}) with average values between $0.43$ and $0.51$. In the \emph{diverse} cohort ({Qwen-1.5-7B-Chat}, {BioMistral-7B-DARE}, {ANIMA-Nectar-v2}, and {Zephyr-beta-math}), the models have an even tighter range of average values ($0.42$ to $0.46$), but they are good at different things, so on a given prompt the runner-up tends to trail the winner by more ($17.1\%$ versus $7.7\%$). For both cohorts the competitive-floor bound limits the welfare loss to well below the worst case of $1/e \approx 36.8\%$, and the limit is lower for the similar cohort ($24.2\%$ versus $27.4\%$). Actual losses are small in both, but the winner shades much more in the diverse cohort ($19.8\%$ versus $12.3\%$), so the similar cohort loses slightly less welfare ($0.4\%$ versus $1.2\%$). Nevertheless, the bidder diversity here is only mild, as every model is still a credible rival on many prompts. When a "generalist" ({Yi-34B-Chat}) is paired with three weaker, "niche" models ({Qwen-1.5-0.5B-Chat}, {DeepSeek-Coder-6.7B-Instruct}, and {Gemma-2B-it}), the runner-up trails by $35.6\%$ on average, the winner shades by $29.1\%$, the loss bound rises to $29.9\%$, and the welfare loss rises to $5.1\%$. Swapping one niche model for a second generalist is enough to bring the welfare loss back down to $0.8\%$, even though half the cohort is still weak (here the generalists are {SOLAR-10.7B-Instruct} and {Neural-Chat-7B-v3-3}, paired with {Qwen-1.5-0.5B-Chat} and {Gemma-2B-it}). These results support our core intuition that value concentration limits shading and thereby reduces overall welfare loss.\looseness=-1 

\begin{table}[t]
\centering
\small

\setlength{\tabcolsep}{4pt}
\begin{tabular}{lccc|cc}
\toprule
Cohort & $\E[\max_i V_i]$ & $1 - \E[V_{(2)}]/\E[V_{(1)}]$ & Loss bound & Shading & Welfare loss (adv) \\
\midrule
Generalists (similar)         & $0.557$ & $7.7\%$  & $24.2\%$ & $12.3\%$ & $0.4\%$ \\
Generalists (diverse)         & $0.553$ & $17.1\%$ & $27.4\%$ & $19.8\%$ & $1.2\%$ \\
\midrule
Generalists (1) + niche (3)   & $0.552$ & $35.6\%$ & $29.9\%$ & $29.1\%$ & $5.1\%$ \\
Generalists (2) + niche (2)   & $0.555$ & $15.9\%$ & $27.9\%$ & $22.3\%$ & $0.8\%$ \\
\bottomrule
\end{tabular}
\caption{Four cohorts constructed to have comparable $\E[\max_i V_i]$, but different value structure. We report welfare as a loss, i.e., as the fraction of optimal welfare that is lost. \emph{Loss bound} is one minus the competitive-floor bound of \Cref{thm:multi_floor} (worst case is $1/e \approx 36.8\%$), and \emph{Welfare loss} is one minus the realized PoA (at the adversarially-chosen equilibrium). Shading is $1 -\E[B_{\mathrm{winner}}/V_{\mathrm{winner}}]$.\looseness=-1}
\label{tab:cohorts}
\vspace{-10pt}
\end{table}

\section{Related work}

Prior work on agentic markets has focused on dominant strategy incentive-compatible (DSIC) mechanisms, where bidding one's true value (or cost, in a procurement auction) is optimal. \citet{chang2020decentralized} first used second-price auctions to solve general, sequential Markov decision problems (though not with modern AI agents), while more recent work has used Vickrey-Clarke-Groves (VCG) auctions to orchestrate LLM agents~\citep{liu2026markets} and reverse second-price auctions to allocate test-time compute~\citep{velasco2026test}. These mechanisms, however, do have significant weaknesses---such as being susceptible to collusion~\citep{robinson1985collusion} and requiring bidders to trust the auctioneer~\citep{rothkopf1990vickrey, asubel2004lovelybutlonely}. In this paper we study the viability of the first-price auction, and argue that its lack of truthfulness is less of a concern under value structures we expect agentic markets to have, such as competitive floors and weakly dependent values.\looseness=-1

The price of anarchy (PoA) of first-price auctions is well studied. Smoothness arguments~\citep{RoughgardenPOA} give a universal $1 - 1/e \approx 0.632$ bound under arbitrary correlation~\citep{SyrgkanisTardos2013, Christodoulou2016}. For independent private values, the bound was later raised to $0.743$~\citep{hoy2018tighter}, and then to $1 - 1/e^2 \approx 0.865$, which is tight~\citep{JinLu2022}. All of these bounds are worst case over the value distribution, and they remain tight even for the best equilibrium~\citep{JinLu2023PoS}. We show how these bounds \emph{can} improve under additional structure on the values. Related settings have their own beyond-worst-case results. For example, when a bidder's value depends on what others privately know, the PoA of simple auctions can be poor, but improves when information is spread evenly~\citep{EdenFeldmanTalgamCohenZviran2021}. Revenue guarantees can also be shown to be robust to small changes in the prior~\citep{BrustleCaiDaskalakis2020}, similar in spirit to our weak dependence result (\Cref{sec:near_independence}), where the $1-1/e^2$ bound degrades smoothly with distance from independence even as equilibrium strategies shift. Finally, our experiments connect to empirical approaches to the PoA~\citep{hoy2015price,hartline2026}, and to average-case analyses that weight equilibria by how likely learning dynamics are to reach them~\citep{panageas2016average}.\looseness=-1

\section{Conclusion}
This paper analyzes the welfare of first-price auctions (FPAs) among AI agents, where the main objection to  non-truthful mechanisms like the FPA is that they are vulnerable to welfare loss from bid shading. Worst-case price of anarchy bounds are pessimistic, but agentic markets often have structure that narrows the gap. We prove that the welfare guarantee rises toward full efficiency when the strongest few bidders have similar values, and that the $1-1/e^2$ bound for independent values degrades smoothly when values are only weakly dependent (e.g., given the query). Finally, we show that platforms themselves can take simple steps to protect the efficiency of their markets, such as filtering out weak bidders with routing rules, or entering their own models to set competitive reserve prices.\looseness=-1

\bibliography{references}
\bibliographystyle{abbrvnat}
\clearpage
\appendix
\setcounter{theorem}{0}
\numberwithin{theorem}{section}
\addcontentsline{toc}{chapter}{Appendices}
\etocsetnexttocdepth{2}
\localtableofcontents
\clearpage
\section{Limitations and future work}
We assume that models and users are aligned in the sense that positive utility derives only from providing value to the user; external sources of utility, such as advertising placement, would obviously require additional analysis. Second, like all price of anarchy results, our bounds only describe bidders who play equilibrium strategies. Auctions composed of agents that are still learning how to bid, or that bid by simple heuristics, could result in realized welfare that is quite a bit worse (or better!) than the welfare at any equilibrium bidding strategy. Extending our analysis to include realistic  learning dynamics for first-price auctions in agentic markets is a valuable next step. Third, we analyze competitive floors and near-independence separately, and in experiments we only compute the competitive-floor bound. Real markets like EmbedLLM appear to have both properties (\Cref{sec:experiments} and \Cref{sec:approx_condt_indep}). Taking the larger of the two guarantees is immediate, but a bound that uses both properties together could be stronger, but is non-trivial. We leave this to future work.

Finally, we focus on welfare in this paper. Revenue matters to the platform and its price of anarchy can behave quite differently~\citep{hartline2014price};  we leave revenue guarantees to future work.\looseness=-1

\section{Additional theoretical results}
\label{sec:additional_theory}

\subsection{Bounds that hold over a family of priors}
\label{sec:priors}
Each bound in \Cref{sec:concentration,sec:near_independence} holds for every equilibrium under a single prior $\mu$. Write $L(\mu)$ for that bound, where each prior can satisfy the assumptions with its own parameters (e.g., its own $\alpha$, $m$, and $\epsilon$ in \Cref{thm:multi_floor}). For a family of priors $\mathcal{M}$, the PoA is then at least the smallest of these bounds:
\begin{equation}
\inf_{\mu\in\mathcal M}\inf_{\mathbf s\in\mathrm{BNE}(\mu)}\frac{\mathbb E_\mu[W(\mathbf s(\mathbf V),\mathbf V)]}{\mathbb E_\mu[V_{(1)}]}\ge \inf_{\mu\in\mathcal M} L(\mu).
\end{equation}
For example, fix $m$, a bound $\bar v>0$, and a distribution function $G$, and let $\mathcal M$ contain every prior with $0<\E_\mu[V_{(1)}]\le\bar v$ under which $V_{(m)}$ stochastically dominates $G$, i.e., $\Prb_\mu(V_{(m)}<t)\le G(t)$ for all $t$. Then for every $\alpha\in[0,1]$ and $\mu\in\mathcal M$,
\[\Prb_\mu\bigl(V_{(m)}<\alpha\E_\mu[V_{(1)}]\bigr)\le\Prb_\mu\bigl(V_{(m)}<\alpha\bar v\bigr)\le G(\alpha\bar v),\]
so every $\mu\in\mathcal M$ has an $(\alpha,m,G(\alpha\bar v))$-competitive floor, and applying \Cref{thm:multi_floor} gives:
\begin{equation}
\inf_{\mu\in\mathcal M}\inf_{\mathbf s\in\mathrm{BNE}(\mu)}\frac{\mathbb E_\mu[W(\mathbf s(\mathbf V),\mathbf V)]}{\mathbb E_\mu[V_{(1)}]}\ge\Bigl(1-\frac1e\Bigr)+\sup_{\alpha\in[0,1]}\alpha\Bigl[\frac1e-S_m\bigl(G(\alpha\bar v)\bigr)\Bigr]_+.
\end{equation}

\subsection{Positively dependent values}
\label{sec:positive_dependence}

In auction theory, positive dependence between bidder values is usually modeled by affiliation~\citep{MilgromWeber1982}. It is interesting to consider what this implies with respect to our bounds in \Cref{sec:concentration}. Informally, values are affiliated if learning that some bidders value a query highly makes it more likely that the others do too, as when all models tend to do well on easy queries and poorly on hard ones. Affiliated values are \emph{associated}, meaning that any two functions of $\mathbf V$ that are nondecreasing in every coordinate are nonnegatively correlated. For associated values, it is straightforward to show that the floor in \Cref{eq:min_max_bound} is always at least as high as it would be if the same values were independent. Let $\mathbf V^*$ have independent coordinates with $V^*_i$ distributed as $V_i$, and write $V^*_{(1)}$ and $V^*_{(n)}$ for its largest and smallest coordinates. Then the following result holds.\looseness=-1

\begin{proposition}
\label{prop:association_min_max}
If $\mathbf V$ is associated and $0<\E[V_{(1)}]<\infty$, then $\E[V_{(n)}]\ge\E[V^*_{(n)}]$ and $\E[V_{(1)}]\le\E[V^*_{(1)}]$. In particular, $\E[V_{(n)}]/\E[V_{(1)}]\ge\E[V^*_{(n)}]/\E[V^*_{(1)}]$.
\end{proposition}
\begin{proof}
\citet{christofides2004connection} show that $\E[f(\mathbf V)]\ge\E[f(\mathbf V^*)]$ for every supermodular $f$, i.e., every $f$ with $f(\mathbf x\vee\mathbf y)+f(\mathbf x\wedge\mathbf y)\ge f(\mathbf x)+f(\mathbf y)$, where $\vee$ and $\wedge$ are the coordinatewise maximum and minimum. Both $\min_i x_i$ and $-\max_i x_i$ are supermodular, since $\min_i(x_i\wedge y_i)=\min\{\min_i x_i,\min_i y_i\}$ and $\min_i(x_i\vee y_i)\ge\max\{\min_i x_i,\min_i y_i\}$, while a symmetric argument applies for the maximum. Applying the result to both functions gives the two inequalities.
\end{proof}

Note, however, that a similar straightforward connection cannot be made to \Cref{thm:multi_floor}! In fact, the bound in \Cref{thm:multi_floor} can get worse, since positive dependence can make bidders fall below the floor on the same queries, which raises $\epsilon$ relative to independent values with the same marginals. We leave a closer study of how dependence between values affects competitive floors to future work.

\subsection{Query-conditional competitive floors}
\label{sec:conditional_welfare_bounds}
The results in the main text apply to a fixed value prior $\mu$. However, as in \Cref{sec:near_independence}, it is reasonable to assume that each query in an agentic market comes with public information $X$ that all bidders see before they bid (such as the prompt text), which induces a conditional prior $\mu_x:=\mu(\cdot\mid X=x)$. As with $\mu$, we assume each $\mu_x$ is common knowledge among the bidders. We can then apply \Cref{thm:multi_floor} conditionally to each query, with its own $\alpha_x$, $m_x$, and $\epsilon_x$, and subsequently average over $X$.

\begin{corollary}
\label{cor:conditional_floor}
Suppose $0<\E[V_{(1)}]<\infty$ and that for almost every $x$, $\mu_x$ has an $(\alpha_x,m_x,\epsilon_x)$-competitive floor, i.e., $\Prb\bigl(V_{(m_x)}\ge\alpha_x\E[V_{(1)}\mid X=x]\mid X=x\bigr)\ge1-\epsilon_x$. Then every $\mathbf s\in\mathrm{BNE}(\mu)$ satisfies
\begin{equation}
\label{eq:conditional_floor_bound}
\E[W(\mathbf s(\mathbf V),\mathbf V)]\ge\E_X\Bigl[\Bigl(\Bigl(1-\frac1e\Bigr)+\alpha_X\Bigl[\frac1e-S_{m_X}\bigl(\epsilon_X\bigr)\Bigr]_+\Bigr)\E[V_{(1)}\mid X]\Bigr].
\end{equation}
\end{corollary}

\begin{proof}
Fix $x$ with $\E[V_{(1)}\mid X=x]>0$. Applying \Cref{thm:multi_floor} under $\mu_x$ gives
\[
\E[W(\mathbf s(\mathbf V),\mathbf V)\mid X=x]\ge\Bigl(\Bigl(1-\frac1e\Bigr)+\alpha_x\Bigl[\frac1e-S_{m_x}\bigl(\epsilon_x\bigr)\Bigr]_+\Bigr)\E[V_{(1)}\mid X=x].
\]
If $\E[V_{(1)}\mid X=x]=0$, both sides are $0$. Taking expectations over $X$ proves the claim.
\end{proof}

The conditional bound improves on applying \Cref{thm:multi_floor} to $\mu$ directly because the floor $\alpha_x\E[V_{(1)}\mid X=x]$ adapts to each query. A single floor $\alpha\E[V_{(1)}]$ is too high for hard queries, which inflates $\epsilon$, and too low for easy ones.

\subsection{Example: Competitive floors from specialist models}
\label{app:specialists}

We give a stylized example that illustrates a competitive floor in a (toy) agentic market. Suppose there are $d$ "capabilities" that models in the market may posses; write model $i$ as having a capability set $A_i\subseteq[d]$. Further suppose that each query requires a set of capabilities $Y\subseteq[d]$ to solve it. A model can "handle" the query if it has every required capability, $Y\subseteq A_i$. Normalize values to $[0,1]$. Then suppose any model that can "handle" the query has value at least $\alpha$ for some $0 < \alpha \leq 1$, and a model that cannot has value $0$. The values of the models that can handle the query may be arbitrarily dependent.\looseness=-1

For a given requirement set $Y$, let $N_Y:=|\{i:Y\subseteq A_i\}|$ be the number of models that can handle it. Since $\E[V_{(1)}]\le1$, on the event $N_Y\ge m$ we have $V_{(m)}\ge \alpha \ge \alpha \,\E[V_{(1)}]$. The market therefore has an $(\alpha ,m,\epsilon)$-competitive floor with $\epsilon=\Prb(N_Y<m)$---i.e., the probability that a random query drawn from the market's query distribution requires a capability combination that fewer than $m$ models  have.\looseness=-1

To see when $\epsilon$ is small, suppose each model has each capability independently with probability $p$, and group queries by the number $T:=|Y|$ of capabilities they require. A query requiring $t$ capabilities can be handled by each model with probability $p^t$, so it has $N_Y\sim\mathrm{Bin}(n,p^t)$ capable models and fails to have $m$ of them with probability $\epsilon_t:=\Prb(\mathrm{Bin}(n,p^t)<m)$. Then $\epsilon=\sum_t\Prb(T=t)\,\epsilon_t$.

As a further simplification, we can write a loose upper bound for $\epsilon$ as follows. Pick a cutoff $t_0\ge0$, and assume that any query requiring $t>t_0$ capabilities always has fewer than $m$ models that can handle it (i.e., $\epsilon_t=1$ for $t>t_0$). Then, since $\epsilon_t$ increases in $t$, we can write
\begin{align*}
\epsilon=\sum_t\Prb(T=t)\,\epsilon_t
&\le\sum_{t=0}^{t_0}\Prb(T=t)\,\epsilon_t+\Prb(T>t_0) \\
&\le\sum_{t=0}^{t_0}\Prb(T=t)\,\epsilon_{t_0}+\Prb(T>t_0) \\
&\le\epsilon_{t_0}+\Prb(T>t_0).
\end{align*}
For example, with $n=100$ models, $m=8$, and $p=0.8$, a query requiring $t_0=9$ capabilities has $13.4$ capable models on average and $\epsilon_{9}=0.03$. So if at most $3\%$ of queries require more than $9$ capabilities, then $\epsilon\le0.06$. Finally, suppose the reward for a "handled" query is always $1$ (i.e., the response is either right or wrong), so $\alpha = 1$. Then \Cref{thm:multi_floor} guarantees at least $(1-1/e)+[1/e-S_8(0.06)]\approx0.89$ of the optimal welfare, compared with $\approx 0.63$ from the $1-1/e$ bound alone. This holds even though the values are dependent through the shared query.

A similar argument  applies when model $i$ has capability $j$ with probability $p_{ij}$, with $\mu_Y=\sum_i\prod_{j\in Y}p_{ij}$. The floor then fails on queries that require a rare capability or many capabilities, where $\mu_Y$ is small, and it holds on everything else as long as several models have high $p_{ij}$ on the required capabilities.\looseness=-1

\section{Additional experimental results}
\label{app:additional_results}
This section gives additional results for the EmbedLLM market experiments described in  \Cref{sec:experiments}. Specifically, \Cref{sec:approx_condt_indep} tests the conjecture from \Cref{sec:near_independence} that bidder values are close to independent once the prompt is known. \Cref{app:platform_bidder_exp,app:router_exp} then empirically evaluate the two platform design choices from \Cref{sec:applications}:  a platform-owned bidder and pre-auction filtering of weak models.

\subsection{Near-independence given the prompt}
\label{sec:approx_condt_indep}

\begin{figure}[t]
  \centering
  \input{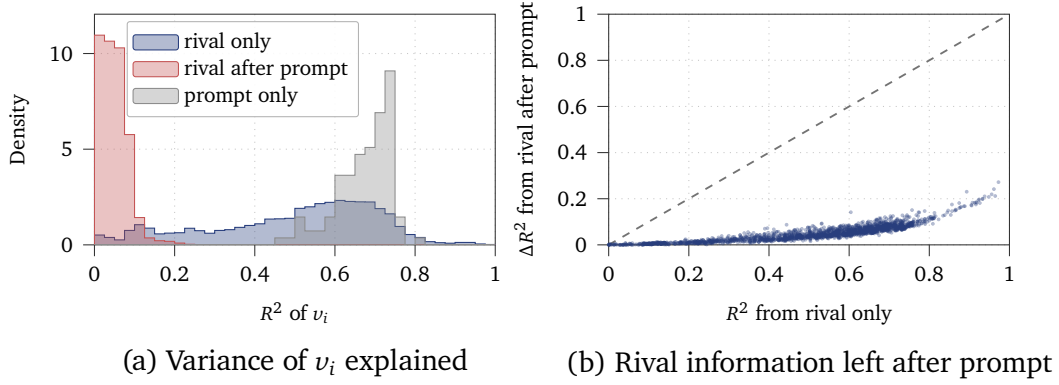}
  \caption{%
    How well the prompt and a rival predict each model's value $v_i$. \textbf{(a)}~Distribution of $R^2$ from the held-out ridge prediction $g_i(q)$ alone (one value per model), and from a rival's value $v_j$ alone and the gain from adding the rival once the prompt is known (one value per ordered model pair $(i,j)$). \textbf{(b)}~Each point is one of $2{,}000$ randomly sampled model pairs $(i, j)$, plotting the rival's $R^2$ on its own against the gain it adds once the prompt is known. The dashed line is $y = x$. For points far below it, most of what the rival's value predicts about $v_i$ is already captured by the prompt.\looseness=-1
  }
  \label{fig:cmi_summary}
\end{figure}

\Cref{sec:near_independence} conjectures that most of the dependence between bidder values in agentic markets can be eliminated by conditioning on the \emph{query}, so that values are close to independent once the query is known. Recall from \Cref{sec:experiments} that $v_i(q) = \max(0, \hat{\E}[\mathrm{acc}_i \mid q] - c_i)$, where the accuracy estimate averages over the $k = 16$ nearest neighbors of $q$. These estimates follow the intuition that a model performs similarly on similar queries. In our construction this holds exactly by design, since every model's value on a prompt is averaged over the same neighbors, which favors the prompt as a predictor here. The general idea still holds more broadly. Much of what makes models succeed or fail together is a property of the query itself, such as its difficulty, its topic, or how clearly it is specified~\citep{hofmann2025fluid,fisch2024stratified,maimon-etal-2026-benchmarks}---and routers already estimate per-query model performance from the prompt~\citep{ong2025routellm,hu2024routerbench,zhuang2025embedllm}.

To test how much of this the prompt alone can predict, we fit a ridge regression $f_i$ from the prompt embedding to $v_i$ for each model, using 5-fold cross-validation. For every ordered pair $(i, j)$ of the 110 models, we then compare how well $v_i$ is predicted by $f_i(q)$ alone, by the rival's value $v_j$ alone, and by both. We measure each fit by its $R^2$, the proportion of the variance of $v_i$ across prompts that the predictors explain in a linear regression. The gain from adding $v_j$ as a feature, i.e., taking $f_i(q) \to f'_i(q, v_j)$, is what rival $j$'s value still tells us about $v_i$ once the prompt is known.\looseness=-1

Results are in \Cref{fig:cmi_summary}. Averaged over all ordered pairs $(i,j)$, the prompt alone gives $R^2 = 0.67$, the rival's value alone gives $R^2 = 0.50$, and adding the rival once the prompt is known raises $R^2$ by only $0.05$ ($0.09$ at the 90th percentile across pairs). The prompt therefore accounts for $90\%$ of what a rival's value explains, and is already the better predictor in $84\%$ of pairs. The pairs with the most dependence left over are near-duplicate models whose values are almost perfectly correlated to begin with.\looseness=-1

To be clear, these results support, but do not imply, the condition behind \Cref{thm:near_independence_welfare}. If values were independent given the prompt, a rival's value would add no information about $v_i$ once the prompt is known, and this is roughly what we see. However, small gains in $R^2$ do not bound the total variation distance $\tau(x)$, which we do not attempt to directly estimate. That said, it is still informative. For example, if we make a simplifying assumption that the remaining dependence is roughly Gaussian, $\tau(x)$ is at most of the order of the correlations between values given $X=x$, which scale with the square root of the gains in $R^2$, up to a factor of at most $\sqrt{n(n-1)/2}$ from summing over pairs.\footnote{To see this, suppose $\mathbf V=\mathbf f(q)+\boldsymbol\varepsilon$ with $\boldsymbol\varepsilon\sim\mathcal N(0,\Sigma)$ independent of the prompt, with correlation matrix $\rho=[\rho_{ij}]$. Then
\[\Delta R^2_{ij}=\frac{\operatorname{Var}(V_i\mid X)-\operatorname{Var}(V_i\mid X,V_j)}{\operatorname{Var}(V_i)}=\rho_{ij}^2(1-R^2_{\text{prompt}}).
\]
Write $P=\mathcal N(\mathbf f(q),\Sigma)$ for this joint distribution and $Q=\mathcal N(\mathbf f(q),\Sigma_0)$ for the product of its marginals, where $\Sigma_0$ is the diagonal of $\Sigma$ (corresponding to the conditionally independent distribution).
Since $\operatorname{tr}(\Sigma_0^{-1}\Sigma)=n$ and $\det\Sigma=\det\Sigma_0\det\rho$, we have
\[
D_{\mathrm{KL}}(P\parallel Q)=\frac12\Bigl[\operatorname{tr}(\Sigma_0^{-1}\Sigma)-n+\ln\frac{\det\Sigma_0}{\det\Sigma}\Bigr]=-\frac12\ln\det\rho\approx\frac12\sum_{i<j}\rho_{ij}^2
\]
for small correlations, so Pinsker's inequality bounds their total variation distance by about $\tfrac12(\sum_{i<j}\rho_{ij}^2)^{1/2}$. Hence $\tau(x)$ is at most of order $(\sum_{i<j}\rho_{ij}^2)^{1/2}\le\sqrt{n(n-1)/2}\,\max_{i<j}|\rho_{ij}|$, where $|\rho_{ij}|=\sqrt{\Delta R^2_{ij}/(1-R^2_{\text{prompt}})}$.}

\subsection{Platform-owned bidder}
\label{app:platform_bidder_exp}

Next, we test the platform participation strategy of \Cref{sec:applications} on randomly sampled markets from EmbedLLM. Each market contains the strongest EmbedLLM model, $n$ rivals drawn at random, and a platform model that bids its own value on every prompt. The platform model is the EmbedLLM model at a given percentile of the ranking by average value, where $0$ is the weakest model and $1$ is the strongest (so at $1$ the platform runs a copy of the best model). Rivals are never drawn from the platform models used at any percentile. This lets us reuse the same random draws of rivals at every percentile, so each curve compares the same markets. For each percentile and $n \in \{5, 10\}$ we draw $10{,}000$ random markets and solve for typical and adversarial equilibria as in \Cref{sec:experiments}.\looseness=-1

Results are shown in~\Cref{fig:platform_bidder_poa}. These markets always include the strongest model and are already efficient, with a mean PoA of about $0.98$, so a weak or average platform model changes welfare very little. The adversarial lower tail, however, improves noticeably as the platform model moves into the top half, from $0.94$ to $0.96$ at $n=5$, and a platform model that matches the best model lifts every curve to about $0.98$--$0.99$ (note that the best \emph{single} model is not the same as the optimal composite model that yields $\E[\max_i V_i]$ welfare). As expected, the larger markets at $n =10$ are already inherently more efficient, and depend less on the platform model to keep bids competitive.\looseness=-1

\begin{figure}[t]
  \centering
\input{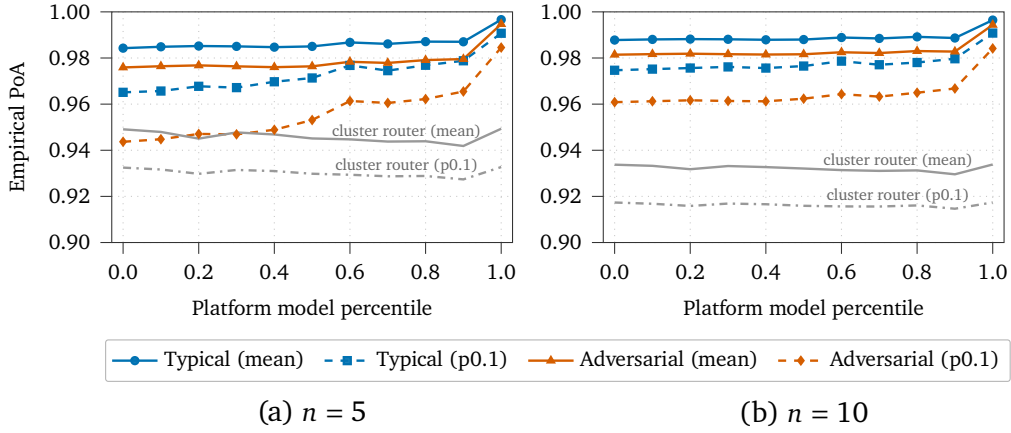}
\begin{tikzpicture}[]
\pgfkeys{/pgf/number format/assume math mode=true}
\begin{axis}[
    name=left,
    width=5.3cm,
    height=3.05cm,
    scale only axis,
    poaaxis,
    tick pos=left,
    xmin=-0.03, xmax=1.03,
    xtick={0, 0.2, 0.4, 0.6, 0.8, 1.0},
    x tick label style={/pgf/number format/.cd, fixed, fixed zerofill, precision=1},
    ymin=0.9, ymax=1.0,
    ytick={0.90, 0.92, 0.94, 0.96, 0.98, 1.00},
    y tick label style={/pgf/number format/.cd, fixed, fixed zerofill, precision=2},
    xlabel={Platform model percentile},
    ylabel={Empirical PoA},
    legend style={at={(0.02,0.03)}, anchor=south west},
]
\addplot[poaline, color=poablue, solid, mark=*, mark options={solid}] coordinates {(0.0,0.98427) (0.1,0.98486) (0.2,0.98520) (0.3,0.98504) (0.4,0.98470) (0.5,0.98503) (0.6,0.98675) (0.7,0.98611) (0.8,0.98710) (0.9,0.98703) (1.0,0.99663)};
\addplot[poaline, color=poablue, dashed, mark=square*, mark options={solid}] coordinates {(0.0,0.96512) (0.1,0.96568) (0.2,0.96773) (0.3,0.96716) (0.4,0.96970) (0.5,0.97134) (0.6,0.97679) (0.7,0.97452) (0.8,0.97686) (0.9,0.97880) (1.0,0.99073)};
\addplot[poaline, color=poaverm, solid, mark=triangle*, mark options={solid}] coordinates {(0.0,0.97591) (0.1,0.97643) (0.2,0.97680) (0.3,0.97638) (0.4,0.97602) (0.5,0.97639) (0.6,0.97835) (0.7,0.97786) (0.8,0.97910) (0.9,0.97954) (1.0,0.99471)};
\addplot[poaline, color=poaverm, dashed, mark=diamond*, mark options={solid}] coordinates {(0.0,0.94368) (0.1,0.94478) (0.2,0.94708) (0.3,0.94690) (0.4,0.94888) (0.5,0.95311) (0.6,0.96135) (0.7,0.96051) (0.8,0.96220) (0.9,0.96548) (1.0,0.98453)};
\addplot[poaline, color=black!40, solid, mark=none, forget plot] coordinates {(0.0,0.94908) (0.1,0.94794) (0.2,0.94505) (0.3,0.94773) (0.4,0.94684) (0.5,0.94512) (0.6,0.94473) (0.7,0.94377) (0.8,0.94390) (0.9,0.94186) (1.0,0.94933)};
\addplot[poaline, color=black!40, dashdotted, mark=none, forget plot] coordinates {(0.0,0.93248) (0.1,0.93166) (0.2,0.92978) (0.3,0.93146) (0.4,0.93098) (0.5,0.92982) (0.6,0.92938) (0.7,0.92871) (0.8,0.92885) (0.9,0.92737) (1.0,0.93272)};
\node[font=\tiny, text=black!60, anchor=south, inner sep=1pt] at (axis cs:0.75,0.94377) {cluster router (mean)};
\node[font=\tiny, text=black!60, anchor=south, inner sep=1pt] at (axis cs:0.75,0.92871) {cluster router (p0.1)};
\end{axis}
\node[anchor=north, yshift=-1.95cm] at (left.south) {(a) $n=5$};
\begin{axis}[
    name=right, at={(left.east)}, anchor=west, xshift=1.2cm,
    width=5.3cm,
    height=3.05cm,
    scale only axis,
    poaaxis,
    tick pos=left,
    xmin=-0.03, xmax=1.03,
    xtick={0, 0.2, 0.4, 0.6, 0.8, 1.0},
    x tick label style={/pgf/number format/.cd, fixed, fixed zerofill, precision=1},
    ymin=0.9, ymax=1.0,
    ytick={0.90, 0.92, 0.94, 0.96, 0.98, 1.00},
    y tick label style={/pgf/number format/.cd, fixed, fixed zerofill, precision=2},
    xlabel={Platform model percentile},
    legend style={at={(-0.1132,-0.42)}, anchor=north, legend columns=-1, /tikz/every even column/.append style={column sep=5pt}},
]
\addplot[poaline, color=poablue, solid, mark=*, mark options={solid}] coordinates {(0.0,0.98779) (0.1,0.98806) (0.2,0.98819) (0.3,0.98811) (0.4,0.98790) (0.5,0.98800) (0.6,0.98888) (0.7,0.98848) (0.8,0.98917) (0.9,0.98867) (1.0,0.99644)};
\addlegendentry{Typical (mean)}
\addplot[poaline, color=poablue, dashed, mark=square*, mark options={solid}] coordinates {(0.0,0.97469) (0.1,0.97518) (0.2,0.97565) (0.3,0.97614) (0.4,0.97563) (0.5,0.97651) (0.6,0.97864) (0.7,0.97708) (0.8,0.97804) (0.9,0.97968) (1.0,0.99081)};
\addlegendentry{Typical (p0.1)}
\addplot[poaline, color=poaverm, solid, mark=triangle*, mark options={solid}] coordinates {(0.0,0.98143) (0.1,0.98168) (0.2,0.98186) (0.3,0.98165) (0.4,0.98150) (0.5,0.98162) (0.6,0.98249) (0.7,0.98215) (0.8,0.98301) (0.9,0.98278) (1.0,0.99436)};
\addlegendentry{Adversarial (mean)}
\addplot[poaline, color=poaverm, dashed, mark=diamond*, mark options={solid}] coordinates {(0.0,0.96085) (0.1,0.96125) (0.2,0.96166) (0.3,0.96137) (0.4,0.96122) (0.5,0.96236) (0.6,0.96432) (0.7,0.96326) (0.8,0.96495) (0.9,0.96678) (1.0,0.98414)};
\addlegendentry{Adversarial (p0.1)}
\addplot[poaline, color=black!40, solid, mark=none, forget plot] coordinates {(0.0,0.93375) (0.1,0.93325) (0.2,0.93178) (0.3,0.93316) (0.4,0.93270) (0.5,0.93208) (0.6,0.93140) (0.7,0.93106) (0.8,0.93127) (0.9,0.92958) (1.0,0.93384)};
\addplot[poaline, color=black!40, dashdotted, mark=none, forget plot] coordinates {(0.0,0.91734) (0.1,0.91681) (0.2,0.91586) (0.3,0.91689) (0.4,0.91660) (0.5,0.91592) (0.6,0.91566) (0.7,0.91561) (0.8,0.91608) (0.9,0.91465) (1.0,0.91738)};
\node[font=\tiny, text=black!60, anchor=south, inner sep=1pt] at (axis cs:0.75,0.93106) {cluster router (mean)};
\node[font=\tiny, text=black!60, anchor=south, inner sep=1pt] at (axis cs:0.75,0.91561) {cluster router (p0.1)};
\end{axis}
\node[anchor=north, yshift=-1.95cm] at (right.south) {(b) $n=10$};
\end{tikzpicture}
  \caption{%
    Empirical PoA with a platform-owned bidder, as a function of where the platform model sits in the ranking of EmbedLLM models by average value. Each point averages $10{,}000$ markets made of the strongest model, $n$ random rivals, and the platform model. Solid lines are means and dashed lines are the $0.1$ percentile across markets, for typical (blue) and adversarial (orange) equilibria. Gray lines mark the cluster router on the same markets (mean solid, $0.1$ percentile dash-dotted).
  }
  \label{fig:platform_bidder_poa}
\end{figure}

\subsection{Pre-auction filtering}
\label{app:router_exp}

Finally, we test pre-auction filtering with a router that admits a model to a prompt's auction only if the model is close to the strongest model on the prompt's cluster. Specifically, model $i$ may bid on cluster $c$ if $\E[V_i \mid c] \ge \beta \max_j \E[V_j \mid c]$, where the maximum runs over the models in the market. Setting $\beta = 0$ admits everyone, and $\beta = 1$ admits only the cluster leader. Because the rule works at the cluster level, it can drop the model with the highest value on a particular prompt. We therefore measure welfare against the first-best of the unfiltered market, and also report the filtered oracle, which is the upper bound of \Cref{eq:router_welfare_bound}. For each $\beta$ and $n \in \{5, 10\}$ we draw $10{,}000$ random markets.\looseness=-1

Results are shown in \Cref{fig:router_poa}. For $\beta \le 0.6$ the router removes only weak models, and the filtered oracle stays near $1$ while mean welfare is unchanged at about $0.99$. The adversarial lower tail, however, \emph{does} improve slightly over this range, from $0.947$ to $0.960$ at $n=5$ and from $0.964$ to $0.968$ at $n=10$. Past $\beta = 0.7$ the router starts dropping models that would have been best on some prompts. Consequently, the filtered oracle falls, and realized welfare falls with it, down to about $0.91$ when only the cluster leader bids. As expected, filtering with a router helps only while it keeps the strongest models. At $\beta = 1$ the auction reduces to a cluster router that sends every prompt to its cluster's best model, and this baseline (gray lines) is below the auction at every $\beta < 1$, in both the mean and the $0.1$ percentile.\looseness=-1

\begin{figure}[t]
  \centering
\input{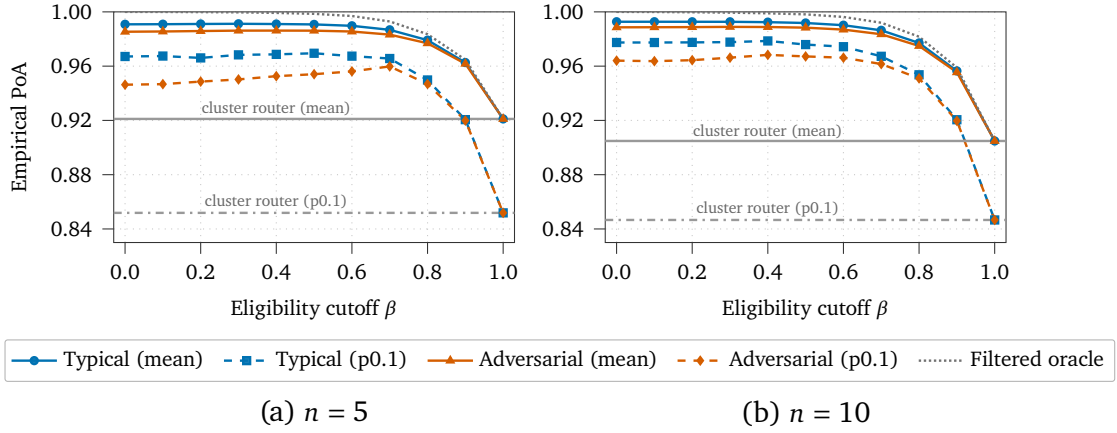}
\begin{tikzpicture}[]
\pgfkeys{/pgf/number format/assume math mode=true}
\begin{axis}[
    name=left,
    width=5.3cm,
    height=3.05cm,
    scale only axis,
    poaaxis,
    tick pos=left,
    xmin=-0.03, xmax=1.03,
    xtick={0, 0.2, 0.4, 0.6, 0.8, 1.0},
    x tick label style={/pgf/number format/.cd, fixed, fixed zerofill, precision=1},
    ymin=0.83, ymax=1.0,
    ytick={0.84, 0.88, 0.92, 0.96, 1.00},
    y tick label style={/pgf/number format/.cd, fixed, fixed zerofill, precision=2},
    xlabel={Eligibility cutoff $\beta$},
    ylabel={Empirical PoA},
    legend style={at={(1.1132,-0.42)}, anchor=north, legend columns=-1, /tikz/every even column/.append style={column sep=5pt}},
]
\addplot[poaline, color=poablue, solid, mark=*, mark options={solid}] coordinates {(0.0,0.99083) (0.1,0.99091) (0.2,0.99108) (0.3,0.99121) (0.4,0.99104) (0.5,0.99075) (0.6,0.98973) (0.7,0.98677) (0.8,0.97915) (0.9,0.96266) (1.0,0.92110)};
\addlegendentry{Typical (mean)}
\addplot[poaline, color=poablue, dashed, mark=square*, mark options={solid}] coordinates {(0.0,0.96711) (0.1,0.96744) (0.2,0.96607) (0.3,0.96832) (0.4,0.96868) (0.5,0.96950) (0.6,0.96737) (0.7,0.96562) (0.8,0.94980) (0.9,0.92057) (1.0,0.85186)};
\addlegendentry{Typical (p0.1)}
\addplot[poaline, color=poaverm, solid, mark=triangle*, mark options={solid}] coordinates {(0.0,0.98536) (0.1,0.98555) (0.2,0.98588) (0.3,0.98618) (0.4,0.98620) (0.5,0.98615) (0.6,0.98552) (0.7,0.98324) (0.8,0.97686) (0.9,0.96176) (1.0,0.92110)};
\addlegendentry{Adversarial (mean)}
\addplot[poaline, color=poaverm, dashed, mark=diamond*, mark options={solid}] coordinates {(0.0,0.94632) (0.1,0.94673) (0.2,0.94862) (0.3,0.95020) (0.4,0.95261) (0.5,0.95406) (0.6,0.95614) (0.7,0.95979) (0.8,0.94682) (0.9,0.91989) (1.0,0.85186)};
\addlegendentry{Adversarial (p0.1)}
\addplot[poaline, color=black!55, densely dotted, mark=none] coordinates {(0.0,1.00000) (0.1,0.99996) (0.2,0.99984) (0.3,0.99965) (0.4,0.99924) (0.5,0.99857) (0.6,0.99695) (0.7,0.99295) (0.8,0.98355) (0.9,0.96473) (1.0,0.92110)};
\addlegendentry{Filtered oracle}
\addplot[poaline, color=black!40, solid, mark=none, forget plot] coordinates {(-0.03,0.92110) (1.03,0.92110)};
\addplot[poaline, color=black!40, dashdotted, mark=none, forget plot] coordinates {(-0.03,0.85186) (1.03,0.85186)};
\node[font=\tiny, text=black!60, anchor=south, inner sep=1pt] at (axis cs:0.4,0.92110) {cluster router (mean)};
\node[font=\tiny, text=black!60, anchor=south, inner sep=1pt] at (axis cs:0.4,0.85186) {cluster router (p0.1)};
\end{axis}
\node[anchor=north, yshift=-1.95cm] at (left.south) {(a) $n=5$};
\begin{axis}[
    name=right, at={(left.east)}, anchor=west, xshift=1.2cm,
    width=5.3cm,
    height=3.05cm,
    scale only axis,
    poaaxis,
    tick pos=left,
    xmin=-0.03, xmax=1.03,
    xtick={0, 0.2, 0.4, 0.6, 0.8, 1.0},
    x tick label style={/pgf/number format/.cd, fixed, fixed zerofill, precision=1},
    ymin=0.83, ymax=1.0,
    ytick={0.84, 0.88, 0.92, 0.96, 1.00},
    y tick label style={/pgf/number format/.cd, fixed, fixed zerofill, precision=2},
    xlabel={Eligibility cutoff $\beta$},
    legend style={at={(0.02,0.03)}, anchor=south west},
]
\addplot[poaline, color=poablue, solid, mark=*, mark options={solid}] coordinates {(0.0,0.99272) (0.1,0.99272) (0.2,0.99270) (0.3,0.99268) (0.4,0.99242) (0.5,0.99180) (0.6,0.99017) (0.7,0.98641) (0.8,0.97707) (0.9,0.95647) (1.0,0.90485)};
\addplot[poaline, color=poablue, dashed, mark=square*, mark options={solid}] coordinates {(0.0,0.97743) (0.1,0.97742) (0.2,0.97752) (0.3,0.97770) (0.4,0.97858) (0.5,0.97590) (0.6,0.97433) (0.7,0.96725) (0.8,0.95363) (0.9,0.92051) (1.0,0.84663)};
\addplot[poaline, color=poaverm, solid, mark=triangle*, mark options={solid}] coordinates {(0.0,0.98862) (0.1,0.98868) (0.2,0.98877) (0.3,0.98893) (0.4,0.98884) (0.5,0.98836) (0.6,0.98694) (0.7,0.98352) (0.8,0.97488) (0.9,0.95548) (1.0,0.90485)};
\addplot[poaline, color=poaverm, dashed, mark=diamond*, mark options={solid}] coordinates {(0.0,0.96405) (0.1,0.96368) (0.2,0.96440) (0.3,0.96624) (0.4,0.96830) (0.5,0.96716) (0.6,0.96622) (0.7,0.96155) (0.8,0.95099) (0.9,0.91969) (1.0,0.84663)};
\addplot[poaline, color=black!55, densely dotted, mark=none] coordinates {(0.0,1.00000) (0.1,0.99994) (0.2,0.99972) (0.3,0.99943) (0.4,0.99894) (0.5,0.99816) (0.6,0.99626) (0.7,0.99200) (0.8,0.98170) (0.9,0.95892) (1.0,0.90485)};
\addplot[poaline, color=black!40, solid, mark=none, forget plot] coordinates {(-0.03,0.90485) (1.03,0.90485)};
\addplot[poaline, color=black!40, dashdotted, mark=none, forget plot] coordinates {(-0.03,0.84663) (1.03,0.84663)};
\node[font=\tiny, text=black!60, anchor=south, inner sep=1pt] at (axis cs:0.4,0.90485) {cluster router (mean)};
\node[font=\tiny, text=black!60, anchor=south, inner sep=1pt] at (axis cs:0.4,0.84663) {cluster router (p0.1)};
\end{axis}
\node[anchor=north, yshift=-1.95cm] at (right.south) {(b) $n=10$};
\end{tikzpicture}
  \caption{%
    Empirical PoA under pre-auction filtering, as a function of the eligibility cutoff $\beta$. Each point averages $10{,}000$ random markets of $n$ models. Solid lines are means and dashed lines are the $0.1$ percentile across markets, for typical (blue) and adversarial (orange) equilibria. The dotted line is the filtered oracle, the welfare of always allocating to the best \emph{admitted} model. Gray lines mark the cluster router, which sends every prompt to its cluster's best model (mean solid, $0.1$ percentile dash-dotted).\looseness=-1
}
  \label{fig:router_poa}
\end{figure}

\section{Equilibrium solver}
\label{app:solver}

We compute approximate Bayes--Nash equilibria with a damped form of fictitious play~\citep{RABINOVICH2013106, heymann2026empiricalstudyfictitiousplay}, implemented directly on the empirical distribution of values. Fictitious play goes back to \citet{brown1951iterative}, and approximating Bayes--Nash equilibria numerically by iterating best responses over a restricted family of strategies is standard for auctions without closed-form solutions~\citep{armantier2008approximation, bosshard2020computing}. Each cluster of prompts is solved on its own. Within a cluster, every prompt is one draw from the joint distribution of the bidders' values, so the solver keeps whatever correlation the data has and never replaces it with a product of marginals or a smoothed prior. Fictitious play has no general convergence guarantee in first-price auctions, but \citet{heymann2026empiricalstudyfictitiousplay} find that it converges numerically to approximate equilibria across a wide range of first-price auctions with correlated values, which is our setting.\looseness=-1

\textbf{Strategies.} Let $\mathcal{Q}_c$ be the set of prompts in cluster $c$. For each bidder $i$, we split its values on $\mathcal{Q}_c$ into $K = 5$ bins of equal mass, and write $k_i(q) \in \{1, \ldots, K\}$ for the bin that contains $v_i(q)$. We define a bidder's strategy to be one shading factor per bin, so on prompt $q \in \mathcal{Q}_c$ bidder $i$ bids $b_i(q) = \gamma_{i,c,k_i(q)}\, v_i(q)$. Such bids never result in over-bidding. The bidder knows the prompt's cluster and its own value, and it knows its rivals only through the empirical distribution of their values and bids within the cluster.\looseness=-1

\textbf{Best responses.} Fix bidder $i$, cluster $c$, and bin $k$, and let $\mathcal{Q}_c^k = \{q \in \mathcal{Q}_c : k_i(q) = k\}$. For every prompt $q \in \mathcal{Q}_c^k$ we record the highest rival bid $r(q)$, the number $t(q)$ of rivals that bid it, and their average value. When the platform enters its own model (\Cref{app:platform_bidder_exp}), it bids its value $v_0(q)$, and this bid counts as one more rival bid. The platform's bid is public (\Cref{sec:applications}), and any bid below it loses, so a bidder that values the prompt above the platform's bid never bids below it. Writing $\underline{v}_0(q) = v_0(q) + 2 \cdot 10^{-6}$, the bid of factor $\gamma$ on prompt $q$ is
\begin{equation*}
b_i^\gamma(q) = \begin{cases} \max\{\gamma\, v_i(q),\, \underline{v}_0(q)\} & \text{if } v_i(q) > \underline{v}_0(q),\\ \gamma\, v_i(q) & \text{otherwise,}\end{cases}
\end{equation*}
which never exceeds $v_i(q)$. Without the platform, $b_i^\gamma(q) = \gamma\, v_i(q)$. If bidder $i$ bids $b$ on prompt $q$, it wins when $b > r(q)$. Ties are split randomly among the $t(q) + 1$ tied bidders. As the winner pays its bid,  bidder $i$'s utility is
\begin{equation}
u_q(b) = \bigl(v_i(q) - b\bigr) \left( \mathbf{1}[b > r(q)] + \frac{\mathbf{1}[b = r(q)]}{t(q) + 1} \right).
\end{equation}
Bidder $i$'s average utility over $\mathcal{Q}_c^k$ from factor $\gamma$ is $$U(\gamma) = \frac{1}{|\mathcal{Q}_c^k|} \sum_{q \in \mathcal{Q}_c^k} u_q(b_i^\gamma(q)).$$ The welfare $W(\gamma)$, meanwhile, is the average over $\mathcal{Q}_c^k$ of the winner's expected value, which is bidder $i$'s value when it wins and the average value of the top rivals when it loses. We search over the grid $\gamma \in \{0, 0.01, \ldots, 1\}$ together with the factors $\gamma = (r(q) + 10^{-6}) / v_i(q) \le 1$ that just outbid the rivals on some prompt $q \in \mathcal{Q}_c^k$. For the \emph{typical} equilibrium, the best response is the factor with the highest utility, and ties go to the factor closest to the current one. For the \emph{adversarial} equilibrium, the best response is the factor with the lowest welfare among those whose utility is positive and at least $95\%$ of the highest utility. If no factor meets both conditions, the bidder keeps its current bids.

\textbf{Updates.} All bidders start by bidding their values. In each round, we visit the bidders in turn, and for each of its bins a bidder moves its bids toward its best response to the current bids of the others,
\begin{equation}
b_i(q) \leftarrow (1 - \eta)\, b_i(q) + \eta\, b_i^{\hat\gamma}(q), \qquad q \in \mathcal{Q}_c^k,
\end{equation}
where $\hat\gamma$ is the best response. Bids start at values, so when the platform's floor applies both terms are at least $\underline{v}_0(q)$ and the damped bid stays above the platform's bid. We use $\eta = 1$ in the first round and $\eta = 0.25$ afterwards; bids are then an exponentially weighted average of the bidder's past best responses. This averaging is what makes the procedure a form of fictitious play, and it helps avoid the cycling that undamped best responses can produce. The solver stops once no bid has moved by more than $10^{-3}$ for $10$ consecutive rounds, or after $100$ rounds. Note that starting from truthful bids means the dynamics mostly shade down, so the solver finds equilibria reachable by lowering bids from truthful bidding and may miss others.\looseness-1

\textbf{Checks.} The typical and adversarial runs share everything except the choice among best responses, so an adversarial equilibrium should never have higher welfare than a typical one on the same market. When it does, the solver has not settled, and we drop the market (for example, this removes at most $9\%$ of markets in the filtering experiment of \Cref{app:router_exp}). The check compares welfare over all of a market's prompts, not cluster by cluster. We also report the $0.1$ percentile of welfare across markets in place of the minimum, so that a few unsettled markets do not drive the results.

\section{Proofs for competitive floors}
\label{app:concentration_proofs}

In this section, we prove results for general mixed (randomized) equilibria $\mathbf{s}$, representing each bid as $B_i = s_i(V_i, Z_i)$ for private seeds $Z_i$ that are independent across bidders and of the values.

Our first lemma establishes a lower bound on welfare anchored at an arbitrary threshold $c \ge 0$. Specifically, let $R = \max_i B_i$, which is the amount the winning bidder pays and is the revenue of the platform. For any constant $c\ge0$, we can bound welfare by the sum of $\mathbb E[\min\{R,c\}]$ and $(1-1/e)\mathbb E[(V_{(1)}-c)_+]$. The first is the expected revenue, capped at $c$. The second is a $(1-1/e)$-fraction of the expected amount by which the top value ($V_{(1)}$) exceeds $c$. When $c=0$, the first term is zero and the lemma recovers the standard $1-1/e$ guarantee~\citep{RoughgardenPOA}.

\begin{lemma}
\label{lem:anchored_welfare}
For every constant $c\ge0$ and any prior $\mu$, every $\mathbf{s} \in \mathrm{BNE}(\mu)$ satisfies
\begin{equation}
\mathbb E[W(\mathbf{s}(\mathbf{V}), \mathbf{V})]
\ge
(1-1/e)\mathbb E[(V_{(1)}-c)_+]
+\mathbb E[\min\{R,c\}].
\label{eq:anchored_welfare}
\end{equation}
\end{lemma}

\begin{proof}
In equilibrium, no bidder can increase its expected utility (i.e., surplus) by a unilateral deviation. That is, for each $i$ and each alternative strategy $s'_i$, $\mathbb E[U_i]\ge\mathbb E[U'_i]$, where $U_i$ is bidder $i$'s utility under $\mathbf{s}$ and $U'_i$ is its utility under $(s'_i,\mathbf{s}_{-i})$, i.e., when $i$ alone switches to $s'_i$ and its rivals keep playing $\mathbf{s}$. We exploit this fact by constructing a deviation $s'_i$ for each bidder and summing the resulting $n$ inequalities. Each inequality is evaluated against the equilibrium rivals; the bidders do not deviate together.

The deviation $s'_i$ is as follows. If $V_i\le c$, bidder $i$ bids zero. If $V_i>c$, it bids $b'_i=c+T_i$, where $T_i$ is drawn independently of everything else with density
\[
f_i(t)=\frac{1}{V_i-c-t},
\qquad
0\le t\le L_i:=(1-1/e)(V_i-c).
\]
This density integrates to one because $\log\bigl((V_i-c)/(V_i-c-L_i)\bigr)=\log e=1$. The deviation also never exceeds $V_i$, so it is never over-bidding.

Let $P_i:=\max_{j\ne i}B_j$ be the highest rival bid, where the $B_j$ are the equilibrium bids. Condition on $V_i=v_i$ and $P_i=p_i$. If $v_i > c$, the deviation wins when $T_i>p_i-c$, and then earns $v_i-c-T_i$. If $v_i\le c$, the deviation bids zero and its utility is zero. Integrating gives
\begin{align*}
\mathbb E[U'_i\mid V_i=v_i,P_i=p_i] 
&\geq \int_{(p_i-c)_+}^{(1- 1/e)(v_i - c)}(v_i-c-t)f_i(t)\,dt  \\
&\ge(1-1/e)(v_i-c)_+-(p_i-c)_+.
\end{align*}
The integral is read as zero when $(p_i-c)_+>(1- 1/e)(v_i - c)$. 

Now sum over bidders. Let $i^*$ be a (measurably chosen) highest-value bidder. Every term is nonnegative and $P_{i^*}\le R$, so conditional on the values and rival bids,
\[
\sum_i\mathbb E[U'_i\mid V_i,P_i]
\ge(1-1/e)(V_{(1)}-c)_+-(P_{i^*}-c)_+
\ge(1-1/e)(V_{(1)}-c)_+-(R-c)_+.
\]
Taking expectations and summing the equilibrium inequalities,
\begin{equation}
\label{eq:expected_utility}
\sum_i \E[U_i] \geq \sum_i \E[U'_i]
= \sum_i\E[\E[U'_i\mid V_i,P_i]] \geq 
(1-1/e)\mathbb E[(V_{(1)}-c)_+]
-\mathbb E[(R-c)_+].
\end{equation}

By definition, $\E[W] =\sum_i\E[U_i] + \E[R].$ Adding expected revenue $\mathbb E[R]$ to \Cref{eq:expected_utility} and using the identity $R-(R-c)_+=\min\{R,c\}$ then proves the result.
\end{proof}

\subsection{Deterministic competitive floor}
We now prove \Cref{proposition:deterministic_poa}, starting with a lower bound on $R$ when competition is deterministic.

\begin{lemma}
\label{lem:winning_bid_floor}
Let $\mu$ be a prior under which $V_{(2)}\ge c$ almost surely for a constant $c\ge0$. Then every
Bayes--Nash equilibrium $\mathbf{s} \in \mathrm{BNE}(\mu)$ satisfies $\max_iB_i\ge c$ almost surely.
\end{lemma}

\begin{proof}
Fix any Bayes--Nash equilibrium $\mathbf{s}$, and let $r_{\min}:=\operatorname{ess\,inf}R$ be the largest constant such that $R\ge r_{\min}$ almost surely. Suppose $r_{\min}<c$. As in Lemma~\ref{lem:anchored_welfare}, write $U_i$ for bidder $i$'s utility under $\mathbf{s}$ and $U'_i$ for its utility under $(s'_i,\mathbf{s}_{-i})$. We now construct a unilateral deviation $s'_i$ for each bidder such that the gains $\mathbb E[U'_i]-\mathbb E[U_i]$ sum to a strictly positive number. Then at least one gain is positive, which contradicts the equilibrium conditions.

Choose a constant $b_\mathrm{dev}$ with $r_{\min}<b_\mathrm{dev}<(c+r_{\min})/2$. Then $\Prb(R<b_\mathrm{dev})>0$ by the definition of $r_\mathrm{min}$, and $\delta:=c+r_{\min}-2b_\mathrm{dev}>0$. Under $s'_i$, bidder $i$ bids $\max\{B_i,b_\mathrm{dev}\}$ when $V_i\ge c$ and keeps $B_i$ otherwise. Note that this deviation is admissible, as it uses only the bidder's own value and equilibrium bid. It also never overbids, since $B_i\le V_i$, and the bid is raised only when $V_i\ge c>b_\mathrm{dev}$.

We compare $U'_i$ with $U_i$ for each bidder on each realized value profile $\mathbf{V} = \mathbf{v}$ and equilibrium winning bid $R =r$. When $r \ge b_\mathrm{dev}$, the deviation never lowers the bidder's utility, so $U'_i \geq U_i$. When $r<b_\mathrm{dev}$, let $A:=\{i:v_i\ge c\}$ be the index set of all bidders with value at least $c$, which has at least two members by assumption. Let $a_i$ be bidder $i$'s probability of winning in equilibrium under the tie rule given $\mathbf{V} = \mathbf{v}$ and $R = r$, so $\sum_{i \in A} a_i\le1$. Each $i\in A$ that unilaterally deviates wins at price $b_\mathrm{dev}$ and earns $v_i-b_\mathrm{dev}$, compared with $a_i(v_i - r)$ in equilibrium. Summing these gains over $A$,
\begin{align*}
\sum_{i \in A} \E[U'_i - U_i \mid \mathbf{V} = \mathbf{v}, R = r] &= \sum_{i\in A}\bigl[(v_i-b_\mathrm{dev})-a_i(v_i-r)\bigr] \\
&=\sum_{i\in A}\bigl[(1-a_i)v_i+a_ir-b_\mathrm{dev}\bigr] \\
&\ge|A|(c-b_\mathrm{dev})-(c-r) \\
&\ge c+r-2b_\mathrm{dev}\\
&\ge c+r_{\min}-2b_\mathrm{dev}=\delta>0,
\end{align*}
using $r\ge r_{\min}$ in the last step. Bidders outside $A$ keep their equilibrium bids, so their gains are zero. Let $D:=\sum_i(U'_i-U_i)$ be the total gain. We have shown $\mathbb E[D\mid\mathbf V,R]\ge\delta$ on $\{R<b_\mathrm{dev}\}$ and $\mathbb E[D\mid\mathbf V,R]\ge0$ otherwise, so
\[
\sum_i\bigl(\mathbb E[U'_i]-\mathbb E[U_i]\bigr)
=\mathbb E\bigl[\mathbb E[D\mid\mathbf V,R]\bigr]
\ge\delta\,\Prb(R<b_\mathrm{dev})>0.
\]
Hence some bidder's expected utility rises by deviating alone. This contradicts the equilibrium conditions of $\mathbf{s}$, so $r_{\min}\ge c$, i.e., $R\ge c$ almost surely.
\end{proof}

\begin{proof}[Proof of~\Cref{proposition:deterministic_poa}]
Recall that $M:=\mathbb E[V_{(1)}]$. Then set $c:=\alpha M$.  Lemma~\ref{lem:winning_bid_floor} gives $R\ge c$ almost
surely, and the hypothesis also gives $V_{(1)}\ge c$ almost surely.  Therefore
Lemma~\ref{lem:anchored_welfare} yields
\[
\mathbb E[W]
\ge
(1-1/e)(M-c)+c
=
(1-1/e)M+ c/e.
\]
Divide by $M$.
\end{proof}

\subsection{Probabilistic competitive floor}
\label{app:bid_shortfall}

We now prove \Cref{thm:multi_floor}. When the competitive floor can fail, the winning bid $R$ can fall below $c$, so Lemma~\ref{lem:winning_bid_floor} no longer applies. Instead, the next lemma bounds the expected shortfall $\E[(c-R)_+]$ when at least $m$ bidders clear $c$ with probability at least $1-\epsilon$. The deviation is the same as in Lemma~\ref{lem:winning_bid_floor}, where each bidder with value at least $c$ raises its bid to a constant $b_\mathrm{dev}$. There, a single $b_\mathrm{dev}$ just above $r_{\min}$ sufficed. Here the deviation can lose when fewer than $m$ bidders clear $c$, so we apply it at every $b_\mathrm{dev}\in(0,c)$ and obtain a differential inequality for $\E[(b_\mathrm{dev}-R)_+]$.

\begin{lemma}
\label{lemma:bid_shortfall}
Let $\mu$ be a prior under which $\Prb(V_{(m)}\ge c)\ge1-\epsilon$ for a constant $c>0$, some $2\le m\le n$, and $\epsilon\in[0,1]$. Then every Bayes--Nash equilibrium $\mathbf{s}\in\mathrm{BNE}(\mu)$ satisfies
\[
\frac{\E[(c-R)_+]}{c}\le S_m(\epsilon).
\]
\end{lemma}

\begin{proof}[Proof of Lemma~\ref{lemma:bid_shortfall}]
Fix any Bayes--Nash equilibrium $\mathbf{s}$ and a constant $b_\mathrm{dev}\in(0,c)$. As before, write $U_i$ for bidder $i$'s utility under $\mathbf{s}$ and $U'_i$ for its utility under $(s'_i,\mathbf{s}_{-i})$, where under $s'_i$ bidder $i$ bids $\max\{B_i,b_\mathrm{dev}\}$ when $V_i\ge c$ and keeps $B_i$ otherwise. This deviation is admissible and never overbids, by the same argument as in Lemma~\ref{lem:winning_bid_floor}. Let $D:=\sum_i(U'_i-U_i)$ be the total gain. In equilibrium, $\E[D]\le0$.

We lower bound $\E[D\mid\mathbf V=\mathbf v,R=r]$ on each realization, as in Lemma~\ref{lem:winning_bid_floor}. Let $A:=\{i:v_i\ge c\}$ and $a_i$ be bidder $i$'s probability of winning in equilibrium under the tie rule given $\mathbf{V} = \mathbf{v}$ and $R = r$, so $\sum_{i \in A} a_i\le1$. When $r\ge b_\mathrm{dev}$, the deviation never lowers any bidder's utility, so the conditional gain is nonnegative. When $r<b_\mathrm{dev}$, the same computation as in Lemma~\ref{lem:winning_bid_floor} gives
\begin{align*}
\E[D\mid\mathbf V=\mathbf v,R=r]
&=\sum_{i\in A}\bigl[(v_i-b_\mathrm{dev})-a_i(v_i-r)\bigr] \\
&\ge|A|(c-b_\mathrm{dev})-(c-r)\\
&=(|A|-1)(c-b_\mathrm{dev})-(b_\mathrm{dev}-r).
\end{align*}
Let $G:=\{V_{(m)}\ge c\}$, so $\Prb(G^c)\le\epsilon$. On $G$ and $\{R < b_\mathrm{dev}\}$ we have $|A|\ge m$, so the conditional gain is at least $(m-1)(c-b_\mathrm{dev})-(b_\mathrm{dev}-R)$. On $G^c$ and $\{R < b_\mathrm{dev}\}$, it is at least $-(b_\mathrm{dev}-R)$, since $|A|\ge1$ gives this from the display above and $|A|=0$ gives zero gain, as there is no deviation. Finally, as mentioned above, on $\{R \geq b_\mathrm{dev}\}$ the gain is non-negative. The three conditional bounds combine into
\[
\E[D\mid\mathbf V,R]\ge(m-1)(c-b_\mathrm{dev})\,\mathbf 1\{R<b_\mathrm{dev}\}\mathbf 1_G-(b_\mathrm{dev}-R)_+.
\]
Let $F(b):=\Prb(R<b)$, $Q(b):=\Prb(\{R<b\}\cap G^c)\le\Prb(G^c)\le\epsilon$, and $H(b):=\E[(b-R)_+]=\int_0^bF(t)\,dt$, where $R\ge0$ by no overbidding, so $H(0)=0$, $H'=F$ almost everywhere, and $H(c)=\E[(c-R)_+]$. Taking expectations and using $\Prb(\{R<b_\mathrm{dev}\}\cap G)=F(b_\mathrm{dev})-Q(b_\mathrm{dev})$ gives
\[
0\ge\E[D]\ge(m-1)(c-b_\mathrm{dev})\bigl[F(b_\mathrm{dev})-Q(b_\mathrm{dev})\bigr]-H(b_\mathrm{dev}).
\]
Since $b_\mathrm{dev}\in(0,c)$ was arbitrary and $Q\le\epsilon$, this gives
\[
H'(b)\le\epsilon+\frac{H(b)}{(m-1)(c-b)}\quad\text{for almost every }b\in(0,c),\qquad H(0)=0.
\]
Multiplying by $(c-b)^{1/(m-1)}$ and integrating gives
\[
H(b)\le\frac{(m-1)\epsilon}{m}\Bigl[c^{m/(m-1)}(c-b)^{-1/(m-1)}-(c-b)\Bigr].
\]
Also, $F\le1$ implies $H(c)\le H(b)+c-b$. Setting $z:=(c-b)/c$, we obtain
\[
\frac{\E[(c-R)_+]}{c}=\frac{H(c)}{c}\le\frac{(m-1)\epsilon}{m}z^{-1/(m-1)}+\Bigl(1-\frac{(m-1)\epsilon}{m}\Bigr)z,\qquad 0<z\le1.
\]
For $0<\epsilon<1$, this expression is minimized at
\[
z_*=\Bigl[\frac{\epsilon}{m-(m-1)\epsilon}\Bigr]^{(m-1)/m}\in(0,1],
\]
and substitution gives $\E[(c-R)_+]/{c} \leq \bigl[\epsilon^{m-1}(m-(m-1)\epsilon)\bigr]^{1/m}=S_m(\epsilon)$. 
When $\epsilon=0$, the integrated inequality gives $H(b)=0$ for every $b<c$, and continuity gives $H(c)=0$. When $\epsilon=1$, the bound follows from $(c-R)_+\le c$.
\end{proof}

\begin{proof}[Proof of Remark~\ref{remark:rate}]Note that the properties in Remark~\ref{remark:rate} follow from the closed form $I_\epsilon(m-1,2)=\epsilon^{m-1}(m-(m-1)\epsilon)$. It equals $0$ at $\epsilon=0$ and $1$ at $\epsilon=1$. Using $m-(m-1)\epsilon\le m$ gives the simple bound. Finally, for $\epsilon\in(0,1)$, write $u:=(1-\epsilon)/\epsilon$. Then $\epsilon^{m-1}(m-(m-1)\epsilon)=\epsilon^m(1+um)$, so $S_m(\epsilon)=\epsilon\,(1+um)^{1/m}$. The factor $(1+um)^{1/m}$ decreases strictly in $m$, because $\log(1+x)>x/(1+x)$ for $x>0$, and tends to $1$ as $m\to\infty$. Hence $S_m(\epsilon)$ decreases strictly to $\epsilon$.
\end{proof}

\begin{proof}[Proof of Theorem~\ref{thm:multi_floor}]
Set $c:=\alpha M$. If $\alpha=0$ or $S_m(\epsilon)\ge1/e$, Lemma~\ref{lem:anchored_welfare} with $c=0$ gives $\E[W]\ge(1-1/e)M$, which is the claim. Otherwise $c>0$, and Lemma~\ref{lemma:bid_shortfall} gives
\[
\E[\min\{R,c\}]=c-\E[(c-R)_+]\ge c\bigl(1-S_m(\epsilon)\bigr).
\]
Since $\E[(V_{(1)}-c)_+]\ge M-c$, Lemma~\ref{lem:anchored_welfare} gives
\[
\E[W]\ge(1-1/e)(M-c)+c\bigl(1-S_m(\epsilon)\bigr)=(1-1/e)M+\alpha M\bigl(1/e-S_m(\epsilon)\bigr).
\]
Divide by $M$.
\end{proof}

\subsection{Welfare with a committed platform bidder}
\label{app:platform_reserve}
We prove \Cref{eq:platform_reserve_poa} for a platform model that publicly announces its value $V_0=v_0(X)$ and commits to bidding it. The other bidders
play a no-overbidding Bayes--Nash equilibrium of the induced auction.
Write $V_{(1)}=\max\{V_0,V_1,\ldots,V_n\}$.

\begin{proof}
Fix $X=x$ and write $v_0=v_0(x)$. Since the platform bids $v_0$, the winning bid satisfies $R\ge v_0$. Run the argument of Lemma~\ref{lem:anchored_welfare} under the conditional law given $X=x$ with $c=v_0$, using deviations only for the strategic bidders; these deviations may depend on $x$ because $x$ is public. On profiles where a strategic bidder has the top value, the argument is unchanged, since the highest rival bid is still at most $R$. On profiles where the platform has the top value, $(V_{(1)}-c)_+=0$, and the required inequality holds because deviation utilities are nonnegative. The platform's utility is always zero, so welfare still equals total bidder utility plus revenue. Hence, using $V_{(1)}\ge v_0$ and $R\ge v_0$,
\[
\mathbb E[W\mid X=x]
\ge
(1-1/e)\left(\mathbb E[V_{(1)}\mid X=x]-v_0\right)+v_0
=
(1-1/e)\,\mathbb E[V_{(1)}\mid X=x]+\tfrac1e\,v_0.
\]
Taking expectations over $X$ and dividing by $\mathbb E[V_{(1)}]$ gives \Cref{eq:platform_reserve_poa}.
\end{proof}

\section{Proofs and additional results for weak dependence}
\label{sec:independence_proofs}

We prove the qualitative Theorem~\ref{thm:near_independence_welfare} by comparing an equilibrium under the conditional joint distribution $\mu_x$ with the same strategies under the product of its marginals, $\nu_x$. Changing the distribution mildly can destroy exact optimality, but it creates only small regret. The main step is to show that small regret under independence also leaves only a small shortfall from the Jin--Lu welfare guarantee. We then transfer this bound back to $\mu_x$. Before moving to the proof of the weak dependence result (Theorem~\ref{thm:near_independence_welfare}), we first record the simpler robustness argument for the $1-1/e$ guarantee, which allows arbitrary correlation~\citep{RoughgardenPOA}. 

\subsection{Warmup: Robustness of the $\big(1-\frac{1}{e}\big)$ bound}

\label{sec:roughgarden_perturb}

Fix a public state $X=x$ and a value distribution $\lambda$. Suppose no bidder
can improve its ex ante expected utility by more than $r$ through a unilateral change of strategy. The proof of Lemma~\ref{lem:anchored_welfare} with $c=0$ gives
\[
\mathbb E_\lambda[W]
\ge (1-1/e)\mathbb E_\lambda[V_{(1)}]-nr.
\]
Each of the $n$ best-response inequalities incurs at most $r$;
summing them and adding revenue proves the bound. The deviations depend
only on each bidder's own value, so the argument allows arbitrary
correlation~\cite[Definition 4.5, Eq. (18), and Theorem 4.6]{RoughgardenPOA}. At exact equilibrium, the $1-1/e$ guarantee therefore holds under every joint distribution. The sharper $1-1/e^2$ guarantee requires a separate argument exploiting independence. 


\subsection{Proof overview}

The proof of Theorem~\ref{thm:near_independence_welfare} has three steps. First, an equilibrium under $\mu_x$ has
regret at most $2\bar v(x)\tau(x)$ when evaluated under $\nu_x$. Second, Proposition~\ref{prop:additive_jin_lu_stability_app} shows that, for a fixed number of bidders, independent-values auctions with vanishing regret cannot retain a positive welfare shortfall below the Jin--Lu guarantee. In its proof, Lemma~\ref{lem:safe_deviations} takes a convergent subsequence of auctions and recovers an exact equilibrium with the same limiting welfare. Finally, using Lemma~\ref{lem:closure_tie_repair} we transfer welfare and optimal welfare from $\nu_x$ back to $\mu_x$. Ties require care in the second step because the winning bidder need not be determined by the limiting bids alone.

\subsection{Welfare shortfall under small regret}
\label{sec:tv_stability}


Write $c_{\mathrm{JL}}:=1-1/e^2$. For a value distribution $\nu$,
define bidder $i$'s ex ante regret by
\[
\operatorname{Reg}_i^\nu(s)
:=\sup_{\widehat s_i}
\mathbb E_\nu\!\left[
u_i(V_i,\widehat s_i(V_i),s_{-i}(V_{-i}))
-u_i(V_i,s_i(V_i),s_{-i}(V_{-i}))\right].
\]
Here $u_i$ is value minus bid on winning and zero otherwise. The alternative strategy may choose a different bid for each own value and may be mixed, but never overbids. Regret measures the average gain available from such a change, with rivals' strategies held fixed.

For values in $[0,1]$, let $\delta_n(r)$ be the supremum additive shortfall below the Jin--Lu guarantee among independent-values auctions with each bidder's regret at most $r$:
\[
\delta_n(r):=
\sup_{\nu,s \in \mathcal{S}_{\nu}(r),\pi}
\left[c_{\mathrm{JL}}\mathbb E_\nu[V_{(1)}]
      -\mathbb E_\nu[W(s)]\right]_+.
\]
The supremum ranges over product distributions on $[0,1]^n$, mixed
no-overbidding profiles satisfying $\max_i\operatorname{Reg}_i^\nu(s)\le r$, and announced measurable bid-only tie rules $\pi$, possibly randomized. All utilities are evaluated using the auction's announced tie rule.

\begin{proposition}
\label{prop:additive_jin_lu_stability_app}
For each fixed $n\ge2$, $\delta_n$ is nondecreasing,
$0\le\delta_n(r)\le1$, and $\delta_n(0)=0$. Moreover,
\begin{equation}
\delta_n(r)\longrightarrow0\qquad\text{as }r\downarrow0.
\label{eq:continuum_defect_vanishes}
\end{equation}
\end{proposition}
Thus sufficiently small regret guarantees a small welfare shortfall
uniformly over all the auctions in the definition. For Theorem~\ref{thm:near_independence_welfare}, set
\begin{equation}
\label{eq:continuum_near_independence_modulus}
 \eta_n^{\mathrm{cont}}(t)
   :=(1+c_{\mathrm{JL}})t+\delta_n(2t).
\end{equation}
The term $\delta_n(2t)$ accounts for the regret created by changing
the distribution; $(1+c_{\mathrm{JL}})t$ accounts for the two welfare
expectations. The proposition implies that $\eta_n^{\mathrm{cont}}$ is nondecreasing,
$\eta_n^{\mathrm{cont}}(0)=0$, and $\eta_n^{\mathrm{cont}}(t)\to0$ as $t\downarrow0$.

\subsubsection{Proof of Proposition~\ref{prop:additive_jin_lu_stability_app}}

For a high-level overview of the proof, suppose small regret could coexist with a fixed positive welfare shortfall. Choose a sequence of such independent-values auctions with regret tending to zero. Compactness gives limiting value--bid distributions, and recording the winner preserves the limiting utilities and welfare. Lemma~\ref{lem:safe_deviations} ensures that feasible deviations still attain nearly the best limiting utility. Lemma~\ref{lem:closure_tie_repair} then shows that the limit can be represented as an exact equilibrium, including a bid-only rule for newly formed ties. The positive shortfall would contradict the Jin--Lu bound.




Let $\Rightarrow$ denote weak convergence of probability measures. A bid just above $t$ beats all rivals bidding at most $t$, regardless of how ties are resolved. We can choose such bids away from atoms of the limiting rival distributions, where winning probabilities converge. A finite partition of the value interval makes these choices into one measurable deviation that never overbids. This is encapsulated in the following Lemma about approximating limiting deviation utility.

\begin{lemma}
\label{lem:safe_deviations}
Fix bidder $i$. Suppose $F^k\Rightarrow F$ are its value distributions
on $[0,1]$ and $H_j^k\Rightarrow H_j$ are its rivals' bid
distributions. In each auction, the value and rival bids are mutually
independent. Let $y_i^k(t)$ be its winning probability from bidding $t$
under that auction's announced bid-only tie rule, and put
\[
G_i(t):=\prod_{j\ne i}H_j(t),\qquad
M_i(v):=\sup_{0\le t\le v}(v-t)G_i(t).
\]
For every $\epsilon>0$, there is a measurable deviation
$0\le d(v)\le v$, the same for every $k$, such that
\begin{equation}
\liminf_{k\to\infty}
\int (v-d(v))y_i^k(d(v))\,dF^k(v)
\ge
\int M_i(v)\,dF(v)-\varepsilon.
\label{eq:safe_deviation_bound}
\end{equation}
\end{lemma}

\begin{proof}
For $0\le w\le v$, bids $t\le w$ increase utility by at most
$v-w$, while bids $w<t\le v$ yield at most $v-w$. Hence $0\le M_i(v)-M_i(w)\le v-w$, and $M_i(0)=0$. Choose a
partition
\[
0=a_0<a_1<\cdots<a_m=1
\]
with mesh at most $\zeta$, where every interior endpoint is a continuity point
of $F$.  Such a partition exists because $F$ has only countably many atoms.

Fix $\ell\ge2$.  The union of the atoms of the finitely many limiting rival
bid distributions $H_j$ is countable.  Right continuity permits us to choose
$t_\ell<a_{\ell-1}$ outside that union such that
\begin{equation}
(a_{\ell-1}-t_\ell)G_i(t_\ell)
\ge M_i(a_{\ell-1})-\zeta.
\label{eq:safe_grid_choice}
\end{equation}
If $M_i(a_{\ell-1})>\zeta$, choose a nearly optimal bid strictly
below $a_{\ell-1}$ and move it slightly upward to a common continuity
point of the $H_j$, still below $a_{\ell-1}$. Right continuity gives
the desired approximation. If $M_i(a_{\ell-1})\le\zeta$, any such
continuity point in $[0,a_{\ell-1})$ suffices because its utility is
nonnegative.

Define $d(v)=0$ on $[0,a_1]$ and $d(v)=t_\ell$ on $(a_{\ell-1},a_\ell]$.  Values at partition endpoints may be assigned to either adjacent interval.  This deviation never overbids.  A bid $t_\ell$ wins whenever every rival bid is strictly below it, regardless of the tie rule, so
\[
y_i^k(t_\ell)
\ge
\prod_{j\ne i}H_j^k(t_\ell-).
\]
For a distribution function $H$, write $H(t-):=\lim_{z\uparrow t}H(z)$. At each selected $t_\ell$, weak convergence gives
$H_j^k(t_\ell-)\to H_j(t_\ell)$. Each partition interval has value-atom-free interior boundaries, so its weighted integral also converges under $F^k\Rightarrow F$. The endpoint $1$ is a boundary of the underlying space $[0,1]$, so an atom there causes no problem.

For $v\in(a_{\ell-1},a_\ell]$, Equation~\eqref{eq:safe_grid_choice} and the
Lipschitz property imply
\[
(v-t_\ell)G_i(t_\ell)
\ge M_i(v)-2\zeta.
\]
On the first interval, $0\le M_i(v)\le\zeta$, while the zero bid has
nonnegative utility.  The left side of~\eqref{eq:safe_deviation_bound} is
therefore at least $\int M_i\,dF-3\zeta$.  Choose $\zeta<\varepsilon/3$.
\end{proof}

\begin{lemma}
\label{lem:closure_tie_repair}
Fix $n$.  Consider a sequence of independent-values first-price auctions on
$[0,1]^n$.  In auction $k$, let $s^k$ be a mixed no-overbidding profile, let
$\pi^k$ be a publicly announced randomized bid-only tie rule, and suppose
\[
\max_i\operatorname{Reg}_i^{\nu^k}(s^k)\longrightarrow0.
\]
There are a subsequence, an independent prior $\nu$, a mixed no-overbidding
profile $s$, and a publicly announced randomized bid-only tie rule $\pi$ such
that $s\in\operatorname{BNE}(\nu)$ and, along the subsequence,
\begin{equation}
\mathbb E_{\nu^k}[V_{(1)}]
\longrightarrow\mathbb E_\nu[V_{(1)}],
\qquad
\mathbb E_{\nu^k}[W(s^k)]
\longrightarrow\mathbb E_\nu[W(s)].
\label{eq:closure_welfare_convergence}
\end{equation}
\end{lemma}

For the limit of the auction itself, we keep each bidder's joint distribution of value and bid, rather than take pointwise limits of
bidding functions. These distributions lie in the compact set $\{(v,b):0\le b\le v\le1\}$. We also record the winner, so utility and welfare pass to the limit even if distinct bids become tied. The remaining issue is to realize that limiting outcome with a rule depending only on bids. We will show that vanishing regret forces any bidder tied at the top with positive surplus to win. Independence then implies that at most one bidder can require this priority at each positive-probability tie level. This allows a bid-only rule that preserves welfare.



\begin{proof}
Let $ D:=\{(v,b)\in[0,1]^2:0\le b\le v\}.$ For bidder $i$ in auction $k$, let $q_i^k$ be the distribution of
$(V_i^k,B_i^k)$ on $D$.  Independence of values and private randomizations implies that the distribution of all value--bid pairs is $q^k:=\bigotimes_iq_i^k$.  Probability measures on the compact metric space
$D$ are weakly compact.  After taking a common subsequence,
\begin{equation}
q_i^k\Rightarrow q_i,
\qquad
q^k\Rightarrow q:=\bigotimes_iq_i.
\label{eq:strategic_measure_convergence}
\end{equation}
Write $F_i$ and $H_i$ for the value and bid marginals of $q_i$.  Disintegrate
$q_i(dv,db)=F_i(dv)s_i(db\mid v)$.  The product of the value marginals is the
limiting independent prior $\nu:=\bigotimes_iF_i$, and closedness of $D$
preserves no overbidding.

Let $I^k\in[n]$ be the winner, including the randomization in $\pi^k$, and let
$\lambda^k$ be the distribution of all value--bid pairs together with $I^k$.  Pass to a
further subsequence for which $\lambda^k\Rightarrow\lambda$.  Its marginal on
$D^n$ is $q$, and
\begin{equation}
\lambda\{b_I=\max_jb_j\}=1,
\label{eq:limit_feasible_winner}
\end{equation}
because the displayed event is closed.  With $[n]$ carrying the discrete
topology, the functions
\[
(v,b,I)\longmapsto(v_i-b_i)\mathbf 1_{\{I=i\}},
\qquad
(v,b,I)\longmapsto v_I
\]
are continuous.  Hence, writing
\[
U_i^*:=\int(v_i-b_i)\mathbf 1_{\{I=i\}}\,d\lambda,
\qquad
W^*:=\int v_I\,d\lambda,
\]
the utilities converge to $U_i^*$ and auction welfare converges to
$W^*$.  Also,
\begin{equation}
\mathbb E_{\nu^k}[V_{(1)}]
\longrightarrow
\int\max_i v_i\,dq
=\mathbb E_\nu[V_{(1)}],
\label{eq:limit_opt_convergence}
\end{equation}
because the maximum is continuous.

For the limiting bid marginals, put
\[
G_i(t):=\prod_{j\ne i}H_j(t),
\qquad
M_i(v):=\sup_{0\le t\le v}(v-t)G_i(t).
\]
Apply Lemma~\ref{lem:safe_deviations} to bidder $i$.  For every $\epsilon>0$, the regret inequality in auction $k$, followed by the lemma and vanishing regret, gives
\begin{equation}
U_i^*\ge\int M_i(v)\,dF_i(v).
\label{eq:limit_lower_utility}
\end{equation}

There is a matching upper bound that uses only feasibility of the limiting
winner.  On the event $I=i$, every rival bid is at most $b_i$.  Since
$v_i-b_i\ge0$ and $\lambda$ has marginal $q$,
\begin{align}
U_i^*
&\le
\int(v_i-b_i)
\mathbf 1_{\{b_j\le b_i\ \forall j\ne i\}}\,dq
\notag\\
&=
\int(v-b)G_i(b)\,dq_i(v,b)
\le
\int M_i(v)\,dF_i(v).
\label{eq:limit_upper_utility}
\end{align}
Every inequality in~\eqref{eq:limit_lower_utility}--
\eqref{eq:limit_upper_utility} is therefore an equality.  In particular,
\begin{equation}
(v-b)G_i(b)=M_i(v)
\quad\text{for }q_i\text{-almost every }(v,b),
\label{eq:onpath_best_response_identity}
\end{equation}
and
\begin{equation}
(v_i-b_i)
\left(
\mathbf 1_{\{b_j\le b_i\ \forall j\ne i\}}
-\mathbf 1_{\{I=i\}}
\right)=0
\quad\lambda\text{-almost surely}.
\label{eq:tie_slack_zero}
\end{equation}
Thus, whenever bidder $i$ has a highest bid and $V_i>B_i$, the limiting
outcome must select $i$, and that ff it lost with positive probability, bidding just above the tie would improve its utility, contrary to vanishing regret.

We now construct an ordinary bid-only tie rule.  A positive-probability tie at
a level $t$ under a product bid distribution requires at least two bid marginals to have
atoms at $t$; there are only countably many such levels.  At any level that can
be the top bid with positive probability, at most one bidder $i$ can satisfy
\begin{equation}
q_i\{B_i=t,\ V_i>t\}>0.
\label{eq:positive_surplus_atom}
\end{equation}
Otherwise two such bidders $i$ and $j$, and the positive-probability event that all remaining bids are at most $t$, would with positive probability both tie at the top and
both have positive surplus.  Equation~\eqref{eq:tie_slack_zero} would require both $I=i$ and $I=j$, an impossibility.

At any level $t$ that can be the top bid with positive probability, call the unique bidder in
\eqref{eq:positive_surplus_atom}, if it exists, $h(t)$.  At a tie for the
highest bid $t$, let the new rule $\pi$ award the item to $h(t)$ if that bidder is
among those tied; otherwise award it to the lowest-index tied bidder.  Away
from ties, award it to the unique highest bidder.  The exceptional levels are
countable, so this is a Borel bid-only rule.

This repair preserves welfare.  At a top tie, every tied bidder other than
$h(t)$ has value exactly $t$ almost surely.  If $h(t)$ is tied and has value
above $t$, Equation~\eqref{eq:tie_slack_zero} says that $\lambda$ already
selects $h(t)$.  In every remaining case all tied bidders have value $t$.
Consequently,
\begin{equation}
\mathbb E_\nu[W(s)]=W^*.
\label{eq:tie_repair_preserves_welfare}
\end{equation}

Finally, under $\pi$, a bidder with a bid $b$ and positive surplus wins
every positive-probability tie at $b$.  Its interim allocation is therefore
$G_i(b)$ for $q_i$-almost every $(v,b)$ with $v>b$.  When $v=b$, its payoff is
zero and~\eqref{eq:onpath_best_response_identity} gives $M_i(v)=0$.  Thus the
on-path expected utility of almost every type $v$ is $M_i(v)$.  Any deviation
to $t$ has interim allocation at most $G_i(t)$ and utility at most
\[
(v-t)_+G_i(t)\le M_i(v).
\]
The same holds for mixed deviations.  Therefore
$s\in\operatorname{BNE}(\nu)$ under $\pi$.  Equations
\eqref{eq:limit_opt_convergence} and
\eqref{eq:tie_repair_preserves_welfare} give
\eqref{eq:closure_welfare_convergence}.
\end{proof}

The argument is uniform over the auctions in the definition, but gives
no explicit rate and does not require equilibrium strategies to vary
continuously with the value distribution.


With all the ingredients in place, we can finally state the main proof of this section. 

\paragraph{Proof of Proposition~\ref{prop:additive_jin_lu_stability_app}}
\begin{proof}
Monotonicity and $0\le\delta_n(r)\le1$ are immediate, and the Jin--Lu theorem gives $\delta_n(0)=0$.  If $\delta_n(r)\not\to0$, choose $\epsilon>0$ and $r_k\downarrow0$
with $\delta_n(r_k)>2\epsilon$. By the definition of the supremum, there are auctions with maximal regret at most $r_k$ and welfare shortfall at least $\epsilon$.
\[
c_{\mathrm{JL}}\mathbb E[V_{(1)}]-\mathbb E[W]\ge\varepsilon.
\]
Lemma~\ref{lem:closure_tie_repair} produces an exact mixed Bayes--Nash
equilibrium of an independent-values first-price auction to which both welfare
terms converge.  The limiting defect remains at least $\varepsilon$, contrary
to the Jin--Lu bound.  Hence~\eqref{eq:continuum_defect_vanishes} holds.
\end{proof}

\subsection{Proof of the main qualitative theorem (Theorem~\ref{thm:near_independence_welfare}).}
\label{app:continuum_stability}

\begin{proof}[Proof of Theorem~\ref{thm:near_independence_welfare}]
Fix a public state $x$ where the hypotheses hold, and abbreviate
$\mu:=\mu_x$, $\nu:=\nu_x$, $\tau:=\tau(x)$, and $\bar v:=\bar v(x)$. If $\bar v=0$, the claim is immediate. Otherwise evaluate the equilibrium strategies under $\nu$ without changing them. The own-value marginals are unchanged, so the same no-overbidding deviations remain available.

For any fixed deviation, its utility and the utility lie in $[0,\bar v]$. Changing the value distribution by TV distance $\tau$
changes each expectation by at most $\bar v\tau$. Equilibrium under
$\mu$ therefore implies
\[
\max_i\operatorname{Reg}_i^\nu(s)\le2\bar v\tau.
\]
After dividing values and bids by $\bar v$, Proposition~\ref{prop:additive_jin_lu_stability_app} gives
\[
\mathbb E_\nu[W]
\ge c_{\mathrm{JL}}\mathbb E_\nu[V_{(1)}]
   -\bar v\delta_n(2\tau).
\]
Since $W,V_{(1)}\in[0,\bar v]$, transferring both expectations back
to $\mu$ yields
\[
\begin{aligned}
\mathbb E_\mu[W]
&\ge \mathbb E_\nu[W]-\bar v\tau\\
&\ge c_{\mathrm{JL}}\mathbb E_\mu[V_{(1)}]
 -\bar v\bigl[\delta_n(2\tau)+(1+c_{\mathrm{JL}})\tau\bigr].
\end{aligned}
\]
This proves \eqref{eq:continuum_near_independence_level_bound} with the error $\eta_n^{\mathrm{cont}}$ as defined in \eqref{eq:continuum_near_independence_modulus}. Applying the argument for almost every public state proves the theorem.
\end{proof}



\subsection{A quantitative bound for weakly-dependent values}
\label{app:non_asymptotic_stability}


Theorem~\ref{thm:near_independence_welfare} gives qualitative stability for arbitrary bounded value distributions. Under additional regularity of the conditional value distribution, its error can be bounded linearly in the total variation distance from independence. The marginals may differ, and the result applies to every no-overbidding Bayes--Nash equilibrium, including mixed equilibria.

Fix $n\ge 2$ and a public state, and normalize the value bound to one. Let $F_i$ be the marginal distribution of $V_i$, put $\nu=\bigotimes_i F_i$, and write
$r=d\mu/d\nu$ and $\varepsilon=\|r-1\|_\infty$.
Assume that each $F_i$ has a continuous density on $[0,1]$ and that
\begin{align}
  0<\ell\le f_i(v)\le u<\infty,\qquad
  &r\in C^1([0,1]^n), \label{eq:quant-prior}\\
  \max_i\|\partial_i r\|_\infty\le L\varepsilon,\qquad
  &0\le\varepsilon\le\frac{1}{4L+1},\qquad L\ge1.
  \label{eq:quant-dependence}
\end{align}
The derivative bound scales with the departure from independence. It prevents $r-1$ from being confined to an arbitrarily small region; at the end of the proof, this allows us to bound $\epsilon$ by a constant times $\tau$. No derivative bounds on the marginal densities are needed.

\begin{theorem}[A quantitative refinement of Theorem~\ref{thm:near_independence_welfare}]
\label{thm:quant-rate}
Under \eqref{eq:quant-prior}--\eqref{eq:quant-dependence}, there is a
constant $C=C(n,\ell,u,L)$, given explicitly at the end of the proof,
such that every no-overbidding Bayes--Nash equilibrium satisfies
\[
  \mathbb E_\mu[W]
  \ge c_{\mathrm{JL}}\mathbb E_\mu[V_{(1)}]-C\tau,
  \qquad \tau:=\|\mu-\nu\|_{\mathrm{TV}}.
\]
\end{theorem}

For the public-state setting of Theorem~\ref{thm:near_independence_welfare}, suppose the conditional distributions of $(V_i/\bar v(x))_{i=1}^n$ satisfy these assumptions whenever $\bar v(x)>0$, with $\ell,u,L$ uniform in $x$. Rescaling does not change total variation, so for almost every $x$,
\[
  \mathbb E[W\mid X=x]
  \ge c_{\mathrm{JL}}\mathbb E[V_{(1)}\mid X=x]
     -\bar v(x)C\tau(x).
\]
The case $\bar v(x)=0$ is immediate. If $\mathbb E[\bar v(X)]<\infty$,
averaging gives
\[
  \mathbb E[W]
  \ge c_{\mathrm{JL}}\mathbb E[V_{(1)}]
     -C\mathbb E[\bar v(X)\tau(X)].
\]

The proof first derives regularity of equilibrium bids, then constructs a nearby exact equilibrium with independent values. We transfer the Jin--Lu welfare guarantee from this equilibrium and express the error in terms of total variation.

\paragraph{Conditional probabilities.}
Put
\[
  \rho=\frac{1+\varepsilon}{1-\varepsilon}\le\frac32,
  \qquad
  \kappa=\frac{L\varepsilon}{1-\varepsilon}\le\frac14.
\]
Since $F_i$ is the actual marginal of $\mu$,
$\int r(v,v_{-i})\,d\nu_{-i}(v_{-i})=1$, differentiating its explicit integral representation is legitimate under the $C^1$ assumption. Thus every nonnegative function $k$ of rival values and private random seeds satisfies, for almost every $v$
\begin{equation}
  (1-\varepsilon)\mathbb E_\nu[k]
  \le \mathbb E_\mu[k\mid V_i=v]
  \le (1+\varepsilon)\mathbb E_\nu[k].
  \label{eq:quant-conditional}
\end{equation}
If $k$ is fixed as $v$ varies and this expectation is positive, its logarithmic derivative has absolute value at most $\kappa$: the derivative is a weighted average of $\partial_i\log r$.

\subsubsection{Proof Overview}
Lemma~\ref{lem:quant-structure} first shows that higher values lead to higher bids and that no single bid amount has positive probability. Lemma~\ref{lem:quant-cutoffs} then bounds the value associated with each
bid and the surplus retained by the bidder. These bounds allow Lemma~\ref{lem:quant-densities} to compare how raising a bid changes the chance of winning under the dependent prior and under independence. Lemma~\ref{lem:quant-comparison} uses these facts to construct a nearby equlibrium with independent values in three steps. First, it determines the required changes in the bidder' bid distributions jointly, accounting for intervals in which a bidder originally places no bids. Second, it bounds the size of these changes using the preceding optimality comparisons. Third, it proves that the constructed bids are globally optimal for their assigned values and controls the resulting change in bid distributions in total variation.  In the proof of
Theorem~\ref{thm:quant-rate}, we apply the Jin--Lu guarantee to this
equilibrium and transfer the welfare bound back to the original
auction. Finally, the derivative assumption bounds the uniform
departure from independence by a constant times the total variation
distance, giving the claimed linear error.

\subsubsection{Proof of Theorem~\ref{thm:quant-rate}}

We state and prove four auxiliary lemmas that will be used for a proof of theorem~\ref{thm:quant-rate}.

\begin{lemma}
\label{lem:quant-structure}
Under the stated prior assumptions, every equilibrium has atomless marginal bid distributions and is outcome-equivalent to a profile of pure
nondecreasing strategies, up to null sets of values. Let $H_i$ be its marginal bid CDF, for $0 \leq b \leq \beta$, define
\begin{equation}
 v_i(b):=F_i^{-1}(H_i(b)),\qquad
 K_i:=\operatorname{supp}(dH_i),\qquad
 \beta:=\max_i\sup K_i.
 \label{eq:regular_cutoffs}
\end{equation}
Then $v_i$ is continuous and nondecreasing, $V_i=v_i(B_i)$ almost surely, and $H_i(b)\ge F_i(b)>0$ for $b>0$. The union of the $K_i$ is $[0,\beta]$, and every $b\in(0,\beta)$ belongs to at least two of them. For $b\in K_i$, $b$ is a best response for the cutoff type $v_i(b)$, including support endpoints interpreted by continuity.
\end{lemma}
Individual supports may have gaps or end below $\beta$. Accordingly, we use a bidder's first-order condition only on its support in $K_i$.

\begin{proof}
\emph{Positive utility and atoms.}
A type $v>0$ can bid $v/2$ and win whenever all rival values are below
$v/2$. No overbidding and \eqref{eq:quant-conditional} imply
\begin{equation}
 U_i(v)\ge(1-\varepsilon)\ell^{n-1}v^n/2^n>0.
 \label{eq:quant-positive-utility}
\end{equation}
Consequently, every sampled equilibrium bid of a positive type has positive surplus. Suppose two or more bidders have atoms at the same positive bid $t$. The event that both bid $t$ and all other values are below $t$ has positive probability. For each bidder with such an atom, consider
the unilateral deviation that raises $t$ to $t+\eta$ whenever $v>t+\eta$. As $\eta\downarrow0$, the sum of these deviations' gains has a strictly positive limit on this top-tie event and a nonnegative limit elsewhere. Bounded convergence contradicts the equilibrium inequalities.

If every bidder has an atom at zero, the same argument applies to the all-zero tie: all values are positive almost surely, and raising a zero bid breaks the tie at a vanishing cost. If some bidder has no zero atom, any other bidder's zero bid wins with probability zero, contrary
to \eqref{eq:quant-positive-utility}. Hence there are no atoms at zero.

Suppose only bidder $i$ has an atom at a positive bid $t$. Put
\[
  u_0=(1-\varepsilon)\ell^{n-1}(t/4)^n,
  \qquad
  p_0=(1-\varepsilon)\Pr(B_i=t)(\ell t/2)^{n-2}>0.
\]
Choose $0<\delta<t/2$ and $\eta>0$ so that $\delta+\eta<p_0u_0/4$. Any rival equilibrium bid in $(t-\delta,t)$ has value at least $t/2$ and surplus at least $u_0$, by the half-value deviation. Raising it to $t+\eta$ is therefore feasible. Conditional on the rival's own value, the event that $i$ bids $t$ and all remaining bidders have values below $t/2$ has probability at least $p_0$. The deviation gains at least $u_0/2$ on this event and loses at most $\delta+\eta$ elsewhere. Its conditional expected gain is strictly positive. Hence no rival places positive mass in $(t-\delta,t)$.

By the preceding argument, no rival can have an atom at $t$. The profitable upward deviations show that rivals also place no mass immediately below $t$. Bidder $i$ can therefore lower its bid at the atom without changing its winning probability. This probability is positive, so the lower payment gives a strict gain, a contradiction.

\emph{Ordering best responses.}
Once atoms have been excluded, conditional winning probabilities
$p_b(v)$ are continuous in $b$ and $v$. Fix $a>b>0$ and let
$p_D(v)=p_a(v)-p_b(v)$. If the higher bid can be optimal then $p_D>0$; otherwise it only increases the payment. Both $p_b$ and a positive $p_D$ are expectations of nonnegative rival-event indicators, so their
logarithmic derivatives are bounded in absolute value by $\kappa$.
For $v>a$,
\begin{equation}
 \frac{d}{dv}\log
 \frac{(v-a)p_D(v)}{(a-b)p_b(v)}
 \ge\frac1{v-a}-2\kappa\ge\frac12.
 \label{eq:regular_single_crossing}
\end{equation}
The higher bid is preferred precisely when this ratio is at least one. Thus, if a lower type weakly prefers $a$ to $b$, every higher type strictly prefers $a$. Bid zero cannot be optimal for a positive type.
It follows that the sets of optimal bids are ordered across types. For each type with more than one optimal bid, choose two such bids and the open interval between them. The ordering of optimal bids makes these intervals disjoint across types. Each contains a rational number, so there are at most countably many such types. The value distributions have
densities, and hence the equilibrium is pure and nondecreasing outside a null set of values.

\emph{Cutoffs and supports.}
Monotone quantile coupling gives $V_i=v_i(B_i)$ almost surely.
Continuity of $F_i^{-1}$ and $H_i$ gives continuity of $v_i$; no
overbidding gives $H_i(b)\ge F_i(b)$. At a support point, approximate $(v_i(b),b)$ by value--bid pairs from the equilibrium. Conditional payoff continuity passes their optimality inequalities to this pair. In particular the cutoff type is strictly above a positive support bid,
by \eqref{eq:quant-positive-utility}.

There can be no open gap in the union of bid supports: a bidder at its right endpoint could lower its bid into the gap and save money without changing its winning probability. The union begins at zero because $H_i(b)\ge F_i(b)>0$ for every $b>0$. Nor can a positive support point belong to just one bidder's support. All of that bidder's rivals would then have constant bid CDFs in a neighborhood, again permitting a profitable small reduction. This proves the remaining assertions.
\end{proof}

The next lemma bounds value cutoffs and surplus uniformly over equilibria. 

\begin{lemma}
\label{lem:quant-cutoffs}
There are constants $Q\ge4$ and $\sigma>0$, depending only on
$n,\ell,u$, such that
\begin{equation}
  (1+\sigma)b\le v_i(b)\le\min\{1,Qb\}
  \quad(0<b\le\beta),
  \qquad \beta\le\frac{1}{1+\sigma}.
 \label{eq:regular_cutoff_bounds}
\end{equation}
In particular, $H_i(b)\le uQb$.
\end{lemma}
\begin{proof}
Fix $0<a\le1/4$ and write $y_i=v_i(a)$,
$y_{\min}=\min_i y_i$, and $y_{\max}=\max_i y_i$. The event $\{\max_i B_i\le a\}$ is, up to null sets, the value rectangle $E\equiv \{V_i\le y_i\ \forall i\}$. Write $R_0=u/\ell$. Under the product distribution conditioned on that rectangle, each of two fixed bidders has probability at least $\ell/(2u)=1/(2R_0)$ of having value at least $y_{\min}/2$. Conditional on the rectangle under $\mu$, the probability of both such values is therefore at least $1/(4\rho R_0^2)\ge1/(6R_0^2)$.
Put $d=y_{\min}/2$. If $d>2a$, let every bidder with value at least $d$ and bid at most $a$ raise that bid to $2a$. These are separate feasible unilateral deviations. Let $E=\{\max_i B_i\le a\}$ and let $N$ count bidders with $V_i\ge d$ and $B_i\le a$. For any event $A$ inside the rectangle $E$, the density bounds give
$\Pr_\mu(A\mid E)\ge\rho^{-1}\Pr_\nu(A\mid E)$. Write $\Delta_i$ for bidder $i$'s gain from its unilateral deviation. On $E$,
\[
  \sum_i\Delta_i
  \ge (d-2a)\mathbf 1_{\{N\ge2\}}-2a.
\]
Indeed, if $N\ge2$, each selected bidder wins when it deviates, whereas the original allocation has only one winner. Outside $E$, any changed bidder originally lost, so its gain is nonnegative.
The equilibrium inequalities consequently give
\[
 (d-2a)\Pr(\text{at least two such values},\ \max_i B_i\le a)
 \le2a\Pr(\max_iB_i\le a).
\]
Using the conditional probability just bounded yields
$y_{\min}\le4a(1+6R_0^2)\le28R_0^2a$. When $d\le2a$, this bound is immediate as well.

Choose $i$ with $y_i=y_{\max}=:t$. Choose equilibrium types v<t tending to $t$; their bids are at most $a$. Their equilibrium utility is at most $v(1+\varepsilon)\prod_{j\ne i}H_j(a)$; their deviation to $v/2$ gives \eqref{eq:quant-positive-utility}. Letting $v\uparrow t$,
\[
 \prod_{j\ne i}H_j(a)
 \ge\frac{\ell^{n-1}t^{n-1}}{\rho2^n}.
\]
At least one of the other cutoffs equals $y_{\min}$, unless all are equal, in which case any other one does. Therefore
\[
 \prod_{j\ne i}H_j(a)
 \le u^{n-1}y_{\min}t^{n-2},
 \qquad
  t\le\rho 2^nR_0^{n-1}y_{\min}
    \le42\cdot2^nR_0^{n+1}a.
\]
For $a>1/4$, the bound $t\le1$ gives $t\le4a$.
We may therefore take
\[
  Q:=2^{n+6}R_0^{n+1},
\]
which proves $v_i(b)\le Qb$ for every positive bid.

For a used positive bid $b$ of type $v$, compare
\eqref{eq:quant-positive-utility} with the upper bound
$(v-b)(1+\varepsilon)(uQb)^{n-1}$ on its utility. Since
$v\ge b$ and $\rho\le2$,
\[
  v-b\ge
  \frac{v^n}{\rho 2^n(R_0Q)^{n-1}b^{n-1}}
  \ge\frac{b}{2^{n+1}(R_0Q)^{n-1}}.
\]
Set
\[
  \sigma:=2^{-(n+1)}(R_0Q)^{-(n-1)}.
\]
Continuity extends this bound to every point of the bidder's support.
Taking support bids tending to $\beta$ gives $\beta\le1/(1+\sigma)$.
Inside a gap of bidder $i$'s support, $v_i$ is constant. If the gap has a
right endpoint in its support, the surplus bound at that endpoint proves
the bound throughout the gap. If its support has already ended, $v_i=1$
and the bound follows from $b\le\beta\le1/(1+\sigma)$. This establishes
\eqref{eq:regular_cutoff_bounds} everywhere on the stated interval.
\end{proof}

\begin{lemma}[Bid densities and first-order conditions]
\label{lem:quant-densities}
The functions $H_i$ and $v_i$ are locally Lipschitz on $(0,\beta]$. At their common differentiability points put
\[
 h_i=H_i'/H_i,\quad S=\sum_i h_i,\quad g_i=S-h_i,
 \qquad \omega_i=\frac1{v_i-b},\quad A(b)=\{i:b\in K_i\}.
\]
For almost every $b\in(0,\beta)$, $S>0$, $h_i=0$ outside $A(b)$, and
\begin{equation}
 \rho^{-1}g_i\le \omega_i\le\rho g_i\quad(i\in A(b)),
 \qquad S\le\frac{2\rho}{\sigma b}.
 \label{eq:regular_active_foc}
\end{equation}
At least two $h_i$ are positive almost everywhere.
\end{lemma}

\begin{proof}
Write
\[
 P_i(t,v)=\Pr_\mu(\max_{j\ne i}B_j\le t\mid V_i=v),
 \qquad G_i(t)=\prod_{j\ne i}H_j(t).
\]
Fix $0<a<\beta$ and a support point $b\ge a$ for bidder $i$. Its cutoff value $v=v_i(b)$ has surplus at least $\sigma a$. Optimality against $t\in[b,b+\sigma a/2]$ gives
\[
 P_i(t,v)-P_i(b,v)
 \le\frac{t-b}{v-t}P_i(b,v)
 \le\frac{2(1+\varepsilon)}{\sigma a}(t-b).
\]
The nonnegative interval-event comparison in
\eqref{eq:quant-conditional} implies
$G_i(t)-G_i(b)\le2\rho(t-b)/(\sigma a)$. For $j\ne i$, positivity of the other factors gives
\[
 H_j(t)-H_j(b)
 \le\frac{2\rho}{\sigma a(\ell a)^{n-2}}(t-b).
\]
Put $M_a:=2\rho/[\sigma a(\ell a)^{n-2}]$ and
$\Delta_a:=\sigma a/2$. Fix $a\le x<y\le\beta$ with $y-x\le\Delta_a$. If $K_j\cap[x,y)$ is empty, continuity gives $H_j(y)=H_j(x)$. Otherwise, let $b$ be the first point of
$K_j\cap[x,y]$. Then $b<y$ and $H_j(b)=H_j(x)$. Another bidder's
support contains $b$, so the preceding increment estimate gives
\[
 H_j(y)-H_j(x)=H_j(y)-H_j(b)
 \le M_a(y-b)\le M_a(y-x).
\]
Partitioning any longer interval into pieces of length at most $\Delta_a$ proves that $H_j$ is Lipschitz on $[a,\beta]$. In particular, it has no singular continuous part there. Since
$F_i^{-1}$ is $1/\ell$-Lipschitz, each $v_i$ is locally Lipschitz on $(0,\beta]$.

We have
\[
  P_i(t,v)
  =\int_{\{z_j\le v_j(t),\ j\ne i\}}
       r(v,z_{-i})\prod_{j\ne i}f_j(z_j)\,dz_{-i}.
\]
At a common differentiability point of the cutoff functions, this
expression is differentiable in $t$ for every fixed $v$. If
$b\in K_i\cap(0,\beta)$, optimality of $b$ for type $v_i(b)$ gives
\[
  \left.\partial_t\log P_i(t,v)\right|_{t=b,\ v=v_i(b)}
  =\frac{1}{v_i(b)-b}=\omega_i(b).
\]
The derivative holds the own value fixed. Boundary differentiation
and the bounds on $r$ give
\[
  (1-\varepsilon)G_i'(b)
  \le \partial_tP_i(b,v_i(b))
  \le (1+\varepsilon)G_i'(b),
\]
as well as
$(1-\varepsilon)G_i(b)\le P_i(b,v_i(b))
\le(1+\varepsilon)G_i(b)$.
If $G_i'(b)=0$, the first-order condition is impossible.
Otherwise, dividing these inequalities proves the first bound in
\eqref{eq:regular_active_foc}. Outside $K_i$, $h_i=0$. Since $|A(b)|=m\ge2$, $\sum_{i\in A}g_i=(m-1)S$. Summing the lower bounds for $\omega_i$ and using $\omega_i\le1/(\sigma b)$ gives $S\le2\rho/(\sigma b)$. If just one $h_i$ were positive, that bidder would have $g_i=0$; if all were zero, every support bidder would have $g_i=0$. Both possibilities contradict the same first-order condition. All statements hold off a common null set.
\end{proof}

The next lemma uses these facts to write an independent comparison equilibrium. 

\begin{lemma}
\label{lem:quant-comparison}
There is $\varepsilon_0=\varepsilon_0(n,\ell,u)>0$ such that, for $\varepsilon\le\varepsilon_0$, there are atomless bid distributions $H_i^*$ with the following properties. Independent draws
$B_i^*\sim H_i^*$, assigned values $V_i^*=v_i(B_i^*)$, induce an exact no-overbidding Bayes--Nash equilibrium with independent values. Moreover,
\begin{equation}
  \sum_i\|H_i-H_i^*\|_{\mathrm{TV}}
  \le\frac{64u n(n+1)Q^2}{\sigma^2}\,\varepsilon.
 \label{eq:regular_bid_tv}
\end{equation}
Here total variation is between probability distributions. The independent value marginals in this comparison auction may differ
from the original $F_i$.
\end{lemma}

\begin{proof}
We keep the value cutoffs $v_i$ and choose new logarithmic bid derivatives. For a bidder placing positive bid density at $b$, the independent first-order condition requires
$\sum_{j\ne i}h_j^*=1/(v_i(b)-b)=\omega_i(b)$.
The following construction satisfies these equalities while allowing a bidder to have zero density. For $b\in(0,\beta]$, order the $\omega_i(b)$ increasingly and define \begin{equation}
 S_*(b)=\min_{2\le k\le n}
          \frac{\sum_{j=1}^k \omega_{(j)}(b)}{k-1},
 \qquad h_i^*(b)=[S_*(b)-\omega_i(b)]_+,
 \qquad z(b)=b+\frac1{S_*(b)}.
 \label{eq:regular_joint_root}
\end{equation}
The functions in \eqref{eq:regular_joint_root} satisfy the following properties of the construction, which we will verify in what follows:
\begin{equation}
 \sum_i h_i^*=S_*,\qquad
 g_i^*:=\sum_{j\ne i}h_j^*=\min\{S_*,\omega_i\},\qquad
 \frac1{Qb}\le S_*\le\frac2{\sigma b}.
 \label{eq:regular_root_properties}
\end{equation}
The function $z$ is continuous and nondecreasing. Almost everywhere, $S_*\le\rho S$. On every open interval disjoint from $K_i$,
\begin{equation}
 [v_i(b)-z(b)]_+\le qQb,
 \qquad q:=1-\rho^{-2}\le4\varepsilon.
 \label{eq:regular_unused_bound}
 \end{equation}
For positive $\omega_1,\ldots,\omega_n$, put
$F(s):=\sum_i(s-\omega_i)_+-s$. At $s=\omega_{(2)}$,
$F(s)=-\omega_{(1)}<0$. Above $\omega_{(2)}$, $F$ is strictly
increasing and tends to infinity, so it has a unique root $s$ there.
For every $k\ge2$,
\[
 (k-1)s-\sum_{j=1}^k\omega_{(j)}
 =\sum_{j=1}^k(s-\omega_{(j)})-s
 \le F(s)=0.
\]
Equality holds when $k$ is the number of coordinates strictly below $s$,
which is at least two. Hence
$s=\min_{2\le k\le n}\sum_{j=1}^k\omega_{(j)}/(k-1)$, as in \eqref{eq:regular_joint_root}. Scaling the coordinates scales the root, and increasing any coordinate decreases $F$ at a fixed $s$, so the root is homogeneous and nondecreasing
in each coordinate. It follows that $\sum_i h_i^*=S^*$ and $g_i^*=\min\{S^*,\omega_i\}$. The bounds in \eqref{eq:regular_root_properties} follow from
$1/(Qb)\le\omega_i\le1/(\sigma b)$. The equation for the common value level is
\begin{equation}
 \sum_i\frac{[v_i(b)-z(b)]_+}{v_i(b)-b}=1.
 \label{eq:regular_level_equation}
\end{equation}
For fixed $z>b$, the left side increases with $b$ and with each $v_i$, and decreases with $z$ near its unique solution. Since each $v_i(b)$ is nondecreasing, $z(b)$ is nondecreasing as well. If the old solution lies below the new bid, the conclusion is immediate. Continuity follows from \eqref{eq:regular_root_properties} and the positive surplus bounds.

At a common differentiability point take $A=A(b)$, of size $m\ge2$. The minimum over all bidders is at most the ratio formed from these $m$ coordinates. Thus \eqref{eq:regular_active_foc} gives
\[
 S_*\le\frac{\sum_{i\in A}\omega_i}{m-1}\le\rho S.
\]
Now let $(l,r)$ be a component of the complement of $K_i$ in $(0,\beta)$, allowing $r=\beta$ for a terminal interval. Necessarily $l>0$: an initial gap would contradict $H_i(b)\ge F_i(b)>0$ and $H_i(0)=0$. The cutoff $v_i$ is constant there, say $w$, and $l$ is an optimal bid for type $w$. By Lemma~\ref{lem:quant-structure}, $l$ is an optimal bid for the cutoff type $w$. For $l<t<r$, optimality implies
\[
 \log P_i(t,w)-\log P_i(l,w)
 \le\log\frac{w-l}{w-t}=\int_l^t\frac{db}{w-b}.
\]
On this interval $h_i=0$, so $g_i=S$. The density comparisons in the proof of Lemma~\ref{lem:quant-densities} apply at every fixed own value, including $w$;
they do not require the bid to lie in $K_i$. Together with $S^*\le\rho S$, they give
$\partial_b\log P_i(b,w)\ge S^*(b)/\rho^2$ almost everywhere. Integrate this inequality, divide by $t-l$, and let $t\downarrow l$. By continuity,
\[
 S_*(l)\le\frac{\rho^2}{w-l},\qquad
 z(l)\ge w-(1-\rho^{-2})(w-l).
\]
Since $z$ is nondecreasing and $w\le Ql$, it follows that $[w-z(b)]_+\le q(w-l)\le qQb$. This proves the gap estimate, including a bidder's terminal interval, where $w=1$.

Next, we need to establish a uniform comparison of the logarithmic derivatives. 

\paragraph{Comparison of derivatives.}
At a common differentiability point, write $A=A(b)$ and
$\eta=\sum_{i\notin A}h_i^*$. If $h_i^*>0$, the identity
$z=b+1/S^*$ gives
$h_i^*\le(S^*)^2[v_i-z]_+$.

For $i\notin A$, the unused-interval estimate gives
\begin{equation}
  \eta\le\frac{4nqQ}{\sigma^2b}
       \le\frac{16nQ\varepsilon}{\sigma^2b}.
\label{eq:regular_inactive_mass}
\end{equation}
We choose $\epsilon_0:=\sigma^2/(32nQ^2)$, so that $\epsilon\le\epsilon_0$ ensures $\eta\le1/(2Qb)$. That is,
\[
  \varepsilon\le\varepsilon_0
     :=\frac{\sigma^2}{32nQ^2}.
\]
Let $S_A$ be the scalar root formed using just the coordinates $\omega_i$ with $i\in A$. Applying the same root to the coordinates $g_i=S-h_i$ gives $S$, since $\sum_{i\in A}h_i=S$ and at least two of those $h_i$ are positive, so $S$ is the root above the second-smallest $g_i$. Homogeneity, monotonicity, and \eqref{eq:regular_active_foc} imply $S/\rho\le S_A\le\rho S$, while adding the other coordinates gives $S_*\le S_A$.

At $t=S_*$, the function $F_A(t)=\sum_{i\in A}(t-\omega_i)_+-t$ equals $-\eta$. There must be at least two $\omega_i$, $i\in A$, strictly below $S_*$. Otherwise $F_A(S_*)$ would be $-S_*$ or $-\omega_i$, at most $-1/(Qb)$, contradicting \eqref{eq:regular_inactive_mass}. Consequently $F_A$ has slope at least one for $t\ge S_*$. Since $F_A(S_A)=0$, we obtain
\[
 0\le S_A-S_*\le\eta,\qquad
 |S_*-S|\le(\rho-1)S+\eta.
\]
For $i\in A$, the positive-part map is $1$-Lipschitz, so
\[
 |h_i^*-h_i|\le |S_*-S|+|\omega_i-g_i|
              \le2(\rho-1)S+\eta.
\]
Adding the inactive coordinates gives
\begin{align}
  \sum_i|h_i^*-h_i|
  &\le2n(\rho-1)S+(n+1)\eta\\
  &\le\left[\frac{15n}{\sigma}
      +\frac{16n(n+1)Q}{\sigma^2}\right]
      \frac{\varepsilon}{b}
   \le\frac{K\varepsilon}{b},
   \label{eq:derivative_estimate}
\end{align}
where we take
\[
  K:=\frac{32n(n+1)Q}{\sigma^2}.
\]
Here $\rho-1\le5\varepsilon/2$ and $S\le2\rho/(\sigma b)$.
Our choice of $\varepsilon_0$ also gives
\[
  K\varepsilon\le\frac{n+1}{Q}\le\frac12.
\]

The factor $1/b$ will be offset by the small-bid probability $H_i(b)\le uQb$ in the coupling below.

As a last step, we need to do an exact independent comparison and calculate its total-variation error.
\paragraph{Exact independent comparison and TV error.}

For $\varepsilon\le\varepsilon_0$ define
\begin{equation}
  H_i^*(b)=\exp\left(-\int_b^\beta h_i^*(t)\,dt\right)
  \quad(0<b\le\beta),\qquad H_i^*(0)=0.
 \label{eq:regular_new_cdfs}
 \end{equation}
Nonnegativity of $h_i^*$ makes \eqref{eq:regular_new_cdfs} nondecreasing, with value one at $\beta$. It is locally absolutely continuous away from zero. Write $d=K\varepsilon\le1/2$. The original distributions also satisfy
$H_i(b)=\exp(-\int_b^\beta h_i)$, and hence
\begin{equation}
 H_i^*(b)\le H_i(b)(\beta/b)^d
           \le uQb(\beta/b)^d\longrightarrow0
 \quad(b\downarrow0).
 \label{eq:regular_endpoint}
\end{equation}
Thus there is no missing mass or atom at zero. Both the original and new distributions have densities on the positive bid interval.

\emph{Global best responses and distintegration.}
The new competing distribution $G_i^*=\prod_{j\ne i}H_j^*$ has
logarithmic derivative $g_i^*=\min\{S_*,\omega_i\}$. Put
\[
 \phi_i^*(b)=b+1/g_i^*(b)=\max\{v_i(b),z(b)\}.
\]
This function is nondecreasing. Where $h_i^*>0$, $\omega_i<S_*$, so $\phi_i^*=v_i$. For a fixed value $v$, expected utility from a positive bid $t$ is $(v-t)G_i^*(t)$, whose derivative is
\[
 G_i^*(t)g_i^*(t)[v-\phi_i^*(t)]
\]
almost everywhere. At $v=\phi_i^*(b)$ it is nonnegative before $b$ and nonpositive afterward. Thus $b$ is a global best response, not just a stationary bid. Bid zero wins with probability zero, and bids above $\beta$ pay more for the same certain allocation. Since a new sampled bid lies where $h_i^*>0$ almost surely, its assigned value is
$v_i(b)=\phi_i^*(b)$. Let $F_i^*$ be the distribution of $v_i(B_i^*)$, and let $s_i^*(v)$ be a regular conditional distribution of $B_i^*$ given $v_i(B_i^*)=v$. The independent value--bid pairs induce independent values with distribution $\nu^*=\bigotimes_i F_i^*$. Almost every sampled bid is a global best response for its assigned value, so the kernels $s_i^*$ form a Bayes--Nash equilibrium under $\nu^*$. Flat parts of $v_i$ may
produce value atoms and mixed conditional bids. The surplus bound ensures no overbidding, and all constructed values lie in $[0,1]$.
\paragraph{Total Variation.}
For a fixed bidder put $m_i=\min\{h_i,h_i^*\}$ and define
\[
  J_i(b)=\exp\!\left(-\int_b^\beta m_i(t)\,dt\right),
  \qquad
  E_i(b)=\exp\!\left(-\int_b^\beta(h_i-m_i)(t)\,dt\right),
\]
and define $E_i^*$ similarly with $h_i^*-m_i$. Extend these CDFs
by their limits at zero and by one above $\beta$. Draw independent
$C_i,D_i,D_i^*$ with these distributions. Then
$\max\{C_i,D_i\}$ has distribution $H_i=J_iE_i$, and
$\max\{C_i,D_i^*\}$ has distribution $H_i^*=J_iE_i^*$.
The two bids agree unless an excess draw exceeds the common draw.
Since $E_i'=(h_i-m_i)E_i\le h_i-m_i$ on $(0,\beta)$,
\[
\begin{aligned}
  \|H_i-H_i^*\|_{\mathrm{TV}}
  &\le \Pr(D_i>C_i)+\Pr(D_i^*>C_i)\\
  &=\int_0^\beta J_i(b)\,[E_i'(b)+(E_i^*)'(b)]\,db\\
  &\le\int_0^\beta J_i(b)\,|h_i(b)-h_i^*(b)|\,db.
\end{aligned}
\]
Atoms of the component distributions at zero contribute nothing to the strict
exceedance events. Also,
\[
 J_i(b)\le H_i(b)
       \exp\left(\int_b^\beta|h_i-h_i^*|\right)
 \le uQb(\beta/b)^d.
\]
Summing over bidders and using
$\sum_i|h_i-h_i^*|\le d/b$ gives:
\[
  \sum_i\|H_i-H_i^*\|_{\mathrm{TV}}
  \le uQd\,\beta^d\int_0^\beta b^{-d}\,db
  =\frac{uQd\beta}{1-d}
  \le2uQK\varepsilon
  =\frac{64u n(n+1)Q^2}{\sigma^2}\,\varepsilon.
\]
This proves \eqref{eq:regular_bid_tv}.
\end{proof}

\paragraph{Proof of Theorem~\ref{thm:quant-rate}}

\begin{proof}
Write
\[
  B:=\frac{64u n(n+1)Q^2}{\sigma^2},
\]
the coefficient obtained in Lemma~\ref{lem:quant-comparison}.
Suppose first that $\varepsilon\le\varepsilon_0$.
Evaluate the original strategies under $\nu$, so their bids have product distribution $\bigotimes_i H_i$, and use the independent comparison equilibrium from that lemma. Let $i^*(b)$ be the winning bidder. Under the original product prior $\nu$, Lemma~\ref{lem:quant-structure} gives independent bids $B_i\sim H_i$ with values $v_i(B_i)$. The comparison equilibrium has independent bids $B_i^*\sim H_i^*$ with the same value cutoffs. Thus the same function
\[
 D(b):=c_{\mathrm{JL}}\max_i v_i(b_i)
       -v_{i^*(b)}(b_{i^*(b)})
\]
represents the welfare shortfall in both auctions.
Both product bid distributions are atomless. The function $D$ takes values in $[c_{\mathrm{JL}}-1,c_{\mathrm{JL}}]$, an interval of length one. The Jin--Lu theorem gives $\mathbb E[D(B^*)]\le0$. Coupling the bids independently across bidders gives
\[
 \left\|\bigotimes_iH_i-\bigotimes_iH_i^*\right\|_{\mathrm{TV}}
 \le\sum_i\|H_i-H_i^*\|_{\mathrm{TV}}\le B\epsilon.
\]

The product bid distributions differ in total variation by at most $B\varepsilon$, so
\[
  c_{\mathrm{JL}}\mathbb E_\nu[V_{(1)}]
     -\mathbb E_\nu[W]\le B\varepsilon.
\]
Changing the value distribution from $\nu$ to $\mu$ changes the expectation of this signed shortfall by at most $\tau$. Consequently,
\begin{equation}
  \mathbb E_\mu[W]
  \ge c_{\mathrm{JL}}\mathbb E_\mu[V_{(1)}]
       -\tau-B\varepsilon.
  \label{eq:quant-intermediate}
\end{equation}
For $\varepsilon>\varepsilon_0$, the same inequality is immediate: $B\varepsilon>B\varepsilon_0=2u(n+1)>c_{\mathrm{JL}}$,
while the signed shortfall is at most $c_{\mathrm{JL}}$. Evaluated at the original strategies, $D(s(V))$ is a bounded function of values with range of length 1. Changing the prior from $\nu$ to $\mu$ therefore changes its expectation by at most $\tau$.

It remains to bound $\varepsilon$ by $\tau$.
Since $\tau=\frac12\int|r-1|\,d\nu$, take a point where $|r-1|=\varepsilon>0$. The derivative bound ensures $|r-1|\ge\varepsilon/2$ within within sup-norm distance $1/(2nL)$ of this point. The intersection of this cube with $[0,1]^n$ has
volume at least $(2nL)^{-n}$ and product density at least $\ell^n$.
Thus
\[
  \tau\ge\frac{\ell^n\varepsilon}{4(2nL)^n},
  \qquad
  \varepsilon\le4\left(\frac{2nL}{\ell}\right)^n\tau.
\]
Substituting in \eqref{eq:quant-intermediate} proves the claim with
\[
  C=1+4B\left(\frac{2nL}{\ell}\right)^n.
\]
At $\varepsilon=0$, the independent-values theorem applies directly.
\end{proof}

The universal $1-1/e$ bound leaves shortfall at most
$(e^{-1}-e^{-2})\mathbb E[V_{(1)}]\le e^{-1}-e^{-2}$
after normalization. Thus on this class of priors we may take
\[
 \eta_n^{\mathrm{reg}}(t)
   :=\min\{e^{-1}-e^{-2},Ct\}.
\]
as the additive-error function in Theorem~\ref{thm:near_independence_welfare} on this class of priors. Importantly, we note that we have not attempted to optimize for the constant.

\paragraph{An example.}
Let $f_i(v)=1+\alpha_i(2v-1), \text{ with } |\alpha_i|\le\frac12,$
and define the joint density relative to the product of these
marginals by
\[
  r(v)=1+\varepsilon
       (2F_1(v_1)-1)(2F_2(v_2)-1),
  \qquad 0\le\varepsilon\le\frac1{13}.
\]
The centered factors preserve the marginals. These priors satisfy the assumptions with $\ell=1/2$, $u=3/2$, and $L=3$, even when the marginals differ. Since $F_i(V_i)$ is uniform under $F_i$, $\|r-1\|_\infty=\varepsilon$ and $\tau=\varepsilon/8$. To see this, observe $ \frac{\epsilon}{2} \mathbb E_\nu\!\left[|2F_1(V_1)-1|\,|2F_2(V_2)-1| \right]=\frac{\epsilon}{8}.$

\end{document}